\documentclass[acmsmall,nonacm]{acmart}
\usepackage[lined,algonl,ruled]{algorithm2e}
\usepackage{listings}
\usepackage{eqparbox}
\usepackage{amsmath}
\usepackage{array}
\usepackage{tikz}
\usepackage{multirow}
\usepackage{array}
\usepackage{makecell}
\usepackage{caption}
\usepackage{enumitem}
\usepackage{ragged2e}
\usepackage{graphicx}
\usepackage{subcaption}

\usetikzlibrary{calc,quotes,arrows.meta}%

\usetikzlibrary{angles,shapes,backgrounds,fit}
\usetikzlibrary{patterns,patterns.meta}

\tikzstyle{ax vertex} = [{circle,inner sep=1.5pt, pattern={Lines[angle=45,distance=1.1pt,line width=0.4pt]}, pattern color=black}]
\tikzstyle{filled vertex}  = [{circle,draw=black, inner sep=1.5pt, fill opacity=0.8}]
\tikzstyle{opac filled vertex}  = [{circle,draw=black, inner sep=1.5pt, fill opacity=0.8, fill opacity=0.5, draw opacity=0.5, dash pattern=on 1.5pt off 1.5pt}]
\tikzstyle{empty vertex}  = [{circle, fill = white, inner sep=1.5pt, minimum width=1.5pt}]

\DeclareMathOperator*{\argmax}{argmax}
\DeclareMathOperator*{\eri}{eri}
\DeclareMathOperator*{\rpi}{rpi}
\DeclareMathOperator*{\hash}{hash}

\definecolor{revblue}{RGB}{40,100,180}

\newcommand{\wzx}[1]{#1}

\newcommand{\ulat}{\mathit{ulat}}
\newcommand{\ulon}{\mathit{ulon}}

\AtBeginDocument{%
  }

\begin{document}

\title{Redundant Array Computation Elimination}

\author{Zixuan Wang}
\orcid{0009-0003-8155-9446}
\affiliation{%
  \institution{Institute of Computing Technology, Chinese Academy of Sciences}
  \city{Beijing}
  \country{China}
}
\affiliation{%
  \institution{University of Chinese Academy of Sciences}
  \city{Beijing}
  \country{China}
}
\email{wangzixuan22@mails.ucas.ac.cn}

\author{Liang Yuan}
\authornote{Corresponding author.}
\orcid{0000-0003-3406-2907}
\affiliation{%
  \institution{Institute of Computing Technology, Chinese Academy of Sciences}
  \city{Beijing}
  \country{China}
}
\email{yuanliang@ict.ac.cn}

\author{Xianmeng Jiang}
\orcid{0000-0002-1664-7402}
\affiliation{%
  \institution{Institute of Computing Technology, Chinese Academy of Sciences}
  \city{Beijing}
  \country{China}
}
\affiliation{%
  \institution{University of Chinese Academy of Sciences}
  \city{Beijing}
  \country{China}
}
\email{jiangxianmeng23z@ict.ac.cn}

\author{Kun Li}
\orcid{0000-0002-1013-1325}
\affiliation{%
 \institution{Institute of Computing Technology, Chinese Academy of Sciences}
  \city{Beijing}
  \country{China}
}
\email{likungw@gmail.com}

\author{Junmin Xiao}
\orcid{0000-0003-0457-4709}
\affiliation{%
  \institution{Institute of Computing Technology, Chinese Academy of Sciences}
  \city{Beijing}
  \country{China}
}
\email{xiaojunmin@ict.ac.cn}

\author{Yunquan Zhang}
\authornotemark[1]
\orcid{0000-0001-7520-9640}
\affiliation{%
  \institution{Institute of Computing Technology, Chinese Academy of Sciences}
  \city{Beijing}
  \country{China}
}
\email{zyq@ict.ac.cn}


\begin{abstract}
Redundancy elimination is a key optimization direction, and loop nests are the main optimization target in modern compilers.
Previous work on redundancy elimination of array computations in loop nests either targets specific computation patterns or fails to recognize redundancies with complex structures.
This paper proposes RACE (Redundant Array Computation Elimination), a hash-based technique that utilizes a novel two-level scheme to identify the data reuse between array references and the computation redundancies between expressions, enabling hierarchical redundancy detection beyond pattern-specific methods.
It traverses the expression trees in loop nests to detect redundancies hierarchically in linear time and generates efficient code with optimized auxiliary arrays that store redundant computation results. Furthermore, RACE supports the expression reassociation with various aggressive strategies to improve the redundancy opportunities.
Experimental results demonstrate the effectiveness of RACE.
\end{abstract}


\maketitle

\section{Introduction}

Redundancy elimination is one of the key optimizations in modern compilers \cite{compilerbooks,optimizing-compilers}. Generally, redundancies fall into two types, unnecessary data movement and redundant computation. The first type includes peephole optimization, copy propagation, and dead-store elimination. The second type contains lazy code motion, dead code elimination, constant folding, common subexpression elimination, value numbering~\cite{constant-propagation,lazy-code-motion,dead-code-elim,partial-redundancies,valuenumber,su2019redundant}. Most techniques in both types target scalar variables and employ the data-flow analysis.

The loop nest transformation is another critical branch of optimizations. It reschedules the computation in the iteration space of the loop nest, such as loop fusion, loop skewing, loop unrolling, vectorizing, and loop tiling~\cite{loop-fusion,loop-fission,loop-interchange}. The primary purpose of these transformations is to improve the data locality and/or expose loop parallelism. In addition, there are several techniques to eliminate redundancies in loop nests. The loop-invariant code motion and induction variable elimination remove redundant computations~\cite{compilerbooks,advanced-compiler-design}, but they still focus on scalar variables. The scalar replacement and array contraction target arrays, but they only reduce the array reference overhead.

Computation redundancies widely exist in loop nests of real-world applications. Furthermore, the redundancies are often related to arrays and across loop iterations, which is beyond the scope of the above-mentioned classic compiler optimization techniques. Several techniques are developed to exploit this kind of redundancy. Cooper et al. \cite{Cooper+:pact08} propose the first redundancy elimination method for loops. Other methods \cite{Basu+:ipdps15,luporini17,luporini20} target specific computation patterns only. A detailed analysis is present in Section \ref{sec-relatedwork}.

Techniques related to arrays usually rely on data dependence analysis~\cite{lamport-data-dependence,maydan1991efficient,glore}. The data dependence detection is generally undecidable since the addresses can be arbitrary functions. Although most array references are expressed by affine functions of loop index variables, solving the dependence precisely over the integer domain still constitutes an NP-hard problem, Diophantine equations~\cite{banerjee-speedup}.

This paper proposes RACE, a novel redundancy elimination method for array computations across loop iterations.
Unlike previous work, which often utilizes data dependence analysis to explore the redundancy,
RACE distinguishes between the principles of data dependence for loop transformation
and data reuse for redundancy detection.
Then it utilizes a novel two-level scheme to identify the array reference reuse and the expression computation redundancy. Redundancies can be easily discovered by grouping the expressions with the same identification value.

RACE traverses the expression trees in loop nests to detect redundancies hierarchically.
Previous work usually requires pairwise comparisons between expressions, resulting in $O(n^2)$ time complexity for $n$ expressions.
However, RACE can complete the
detection on the expression trees in linear time based on the two-level identification scheme.
It generates efficient redundancy-eliminated codes with optimized auxiliary arrays that store
redundant computing results.
The code preserves the original order in the expression evaluation
by only considering redundancies between binary operations, i.e. binary expression trees.
Thus the result consistency is guaranteed.

Furthermore, RACE supports expression reassociation (rearranging operands to expose more opportunities for redundancy elimination)
with various aggressive strategies
\wzx{to enhance redundancy detection.}
This is equivalent to detecting redundancies on n-ary expression trees,
where the computation order can be reorganized.
This paper makes the following contributions:

\begin{itemize}
\item A novel two-level scheme to identify the array references reuse and the expressions computation redundancy.

\item A linear-time algorithm for detecting hierarchical redundancies in binary expression trees with guaranteed floating-point accuracy.

\item An algorithm supporting expression reassociations for discovering more redundancies on n-ary trees.

\end{itemize}

The remainder of this paper is organized as follows.
Section~\ref{sec-motivation} presents an example to explain
the basic idea of this work.
Section~\ref{sec-relatedwork} reviews related work.
Section~\ref{sec:overview} gives an overview of RACE framework.
Section~\ref{sec:identification} introduces the two-level identification scheme for finding array reference reuse and expression computation redundancy.
Section~\ref{sec:detectionbinary} describes the linear redundancy detection algorithm
on binary expression trees
based on the two-level identification scheme
and the optimization of the auxiliary array organization.
Section~\ref{sec:detectionnary} provides
the enhanced algorithm focusing on n-ary expression trees 
and the simplification of the detection strategy.
Section~\ref{sec:threats} discusses the implementation and limitations of RACE.
Section~\ref{sec:experiment} presents the evaluation
and Section~\ref{sec:conclusion} concludes the paper.

\section{Motivation}
\label{sec-motivation}

\begin{figure}[htbp]
  \centering
  \includegraphics[width=\linewidth, trim=0 2 0 6, clip]{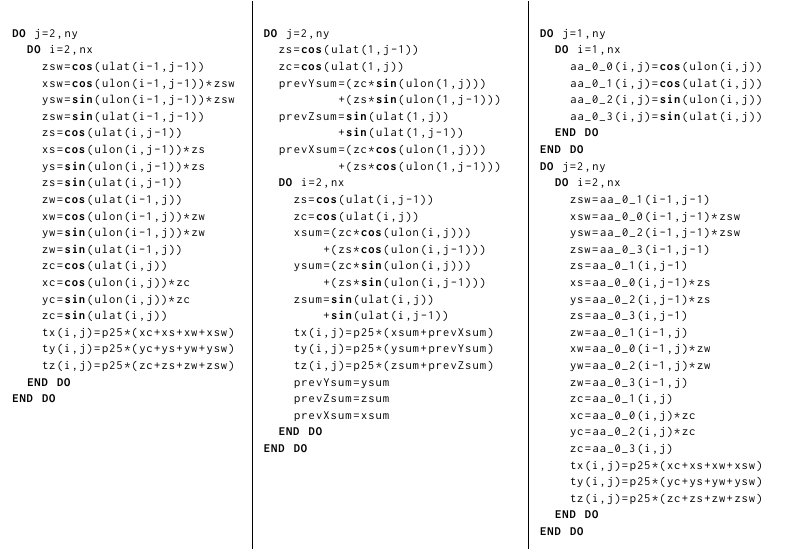}
  \captionsetup{skip=3pt}
  \caption{An example from the Parallel Ocean Program code. Left: the original code. Middle: optimized code by ESR. Right: code generated by RACE after one iteration.}
  \label{fig-pop1}
\end{figure}

\begin{figure}[htbp]
  \centering
  \includegraphics[width=0.9\linewidth, trim=0 2 0 6, clip]{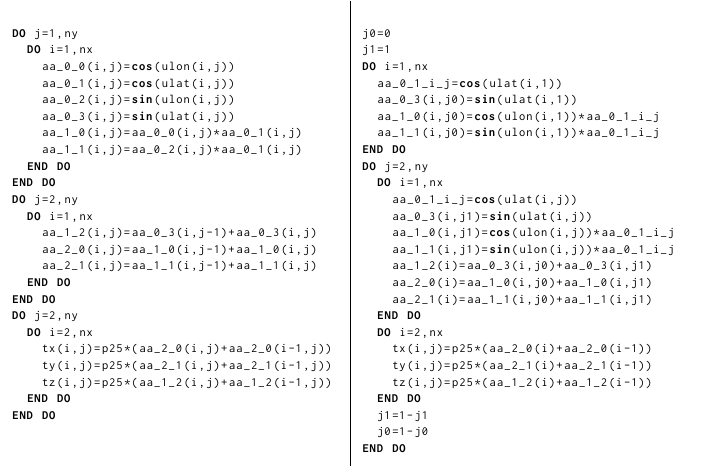}
  \captionsetup{skip=3pt}
  \caption{An example from the Parallel Ocean Program code. Left: code generated by RACE after three iterations. Right: final code by RACE.}
  \label{fig-pop2}
\end{figure}

This section illustrates the fundamental idea of this work through an example.
The left code in Figure~\ref{fig-pop1}
presents a loop nest (in Fortran language) from the Los Alamos National Lab's Parallel Ocean Program (POP).
In each iteration of the loop, it involves $16$ $\sin/\cos$ calls, $11$ multiplications and $9$ additions. 
This original code incorporates the computation redundancy elimination inside the innermost loop.
It employs temporary variables $\mathit{zsw}$, $\mathit{zs}$, $\mathit{zw}$ and $\mathit{zc}$ to hold the corresponding results that are
used twice in successive codes. 
Thus it explicitly saves 4 redundant $\sin/\cos$ calls.
These redundancies can be automatically eliminated by 
classic common subexpression elimination or value numbering in modern compilers.

In addition to the redundancy in the basic block of the loop body, there are obviously other redundancies. 
ESR \cite{Cooper+:pact08} is able to find the redundancies across iterations of the innermost loop.
The code optimized by ESR is shown in the middle code listing of Figure~\ref{fig-pop1}.
ESR identifies that, for example, $\mathit{xw}+\mathit{xsw}$ is the same as the $\mathit{xc}+\mathit{xs}$ computed in the previous iteration of the innermost loop.
Thus it employs extra temporary variables, such as $\mathit{prevXsum}$, to store these redundant computations.
If the three multiplications with the loop-invariant variable $p25$ are ignored, ESR removes half redundancies
in each iteration compared with the original code:
$\sin/\cos$ calls are reduced from 16 to 8, multiplications from 8 to 4, and additions from 9 to 6.

However, computation redundancies also exist across iterations of the outer loop.
For example, the $\mathit{zsw}$ value in iteration $(i,j)$ is the same as $\mathit{zw}$ value in iteration $(i,j-1)$.
Thus it is reasonable to store these values in extra auxiliary arrays.
RACE is able to find all the redundancies in loop nests.
RACE follows a hierarchical scheme.
It detects the redundancies between binary expressions iteratively,
replaces them with auxiliary arrays, and generates the precompute loop before the original loop.
The right list in Figure~\ref{fig-pop1} 
shows the code generated after one iteration of RACE.
It allocates four new auxiliary arrays $aa\_0\_*$, where
$0$ means the first iteration of RACE,
and replaces the computation redundancies with corresponding auxiliary array loads.

RACE repeats the redundancy detection and 
finds all the redundancies after three iterations,
as shown in the left code of Figure \ref{fig-pop2}.
Note that the three auxiliary arrays in the second level
$aa\_1\_*$ belong to two precompute loops since
their ranges differ.
RACE reduces the cost to  $4$ $\sin/\cos$ calls, $5$ multiplications and $6$ additions
per iteration.
The auxiliary arrays exhibit a hierarchical structure
which is used by RACE to optimize the final code, 
as listed in the right code of Figure \ref{fig-pop2}.
Auxiliary arrays $aa\_0\_0$, $aa\_0\_1$ and  $aa\_0\_2$
are only used by other auxiliary arrays in the same precompute loop.
Thus they are removed or replaced with
a scalar variable for multiple uses.
$aa\_0\_3$, $aa\_1\_0$ and  $aa\_1\_1$
can be reduced to arrays whose sizes in $j$ dimension is $2$,
and $aa\_1\_2$, $aa\_2\_0$ and  $aa\_2\_1$
can be reduced to one-dimensional arrays.

\section{Related Work}
\label{sec-relatedwork} 

\subsection{Data Dependence}

The data dependence test 
\cite{optimizing-compilers,yu2012fast,data-depend-83} targets two array references and at least one of them is a write. On the contrary, redundant computation means two or more expressions performing the same operation on the same array elements (or the same values). 
This difference is not critical, since dependence techniques can also be applied to identify two array reads accessing the same memory address. However, the two analyses differ in more essential aspects.

Specifically, the essential differences between the data dependence test and the computation redundancy elimination are three-fold.
First, the dependence test seeks to 
check whether there is at least one iteration that the two references access the same memory address. 
However, the redundancy elimination in loops only targets a large number
of common computations.
Second, the practical general dependence test technique is often conservative in the sense that it may report a false dependence but must report all real data dependence.
However, as an optimization technique, the redundancy elimination should
ensure a potential performance improvement.
Finally, the data dependence test focuses on array references,
while the computation redundancy targets expressions.
Therefore, it needs a new method to determine the redundancy 
between two data reuses and between two expressions.

\subsection{Redundancy Elimination}

Two classic redundancy elimination techniques for the scalar computation
are the value numbering and common subexpression elimination (CSE)~\cite{valuenumber,constant-propagation,common-subexpr-elim}.
The key idea underlying both methods is to record the identification of the operands.
The value numbering assigns the same number to the variables that share the same value,
while CSE utilizes the available expression in the data-flow analysis to recognize the operand change.
However, they only target scalar variables.

Cooper et al.~\cite{Cooper+:pact08} propose the enhanced scalar replacement (ESR), the seminal redundancy elimination method for loops. ESR employs two classic optimization techniques, value numbering and classic scalar replacement (CSR)~\cite{csr}.
ESR has the following disadvantages.
First, ESR considers recomputation only across the innermost loop, thereby failing to discover all redundancies. Second, ESR incorporates scalar replacement, which relies on advanced data dependence analysis. However, data dependence is not suitable for redundancy detection, which requires a more precise definition and method. Third, ESR is a two-step method that separates the dependence-based array partition and expression-level index alignment. This separation leads to unnecessary comparisons between expressions that are not redundant. Fourth, the affinity graph in ESR fails to identify redundant computations among arrays within the same array partition. Finally, ESR seeks to find the maximal computation redundancy in a single execution of the algorithm. However, as the example in Section~\ref{sec-motivation} demonstrates, the computation redundancy may exhibit a hierarchical property.

Deitz et al.~\cite{Deitz+:ics01}
proposed array subexpression elimination (ASE)
for sum-of-product array expressions.
Like ESR, it also aims to eliminate the redundancies over the innermost loop dimension. 
ASE depends on a specific framework called neighborhood tablet 
to search the redundancy.
Basu et al.~\cite{Basu+:ipdps15} targeted stencil computations
and designed a specific optimization scheme. 
It employs extra arrays to 
buffer partial sums within the
leading plane of the whole iteration space.
However, it only focuses on simple stencils and 
reduces redundant computations effectively for 
stencils with symmetric coefficients.
These methods lack generality
and fail to process complex subexpressions, e.g., $ A(i) * C(j,k)$.
Chi and Cong \cite{Chi.Cong:dac20} developed
heuristic search-based reuse (HSBR)
for stencil computations. They did observe the conflict between expressions
but did not formalize it as we do.

RCE (redundant computation elimination) \cite{Yuan+:ica3pp16}
is able to explore redundancies with temporal locality for stencils
but their search algorithm compares all possible expressions
and incurs high overhead.
This scheme is also utilized by the
loop-carried CSE method proposed in \cite{Kronawitter+:ica3pp16}.
On the contrary, our method assigns each expression
a hash number and redundancies can be easily
identified by extracting expressions with the same hash number.
Wen et al.~\cite{Wen+:pact15} implemented a profiler to
pinpoint redundant computations by runtime checking. It finds redundancies in
several stencil computations and eliminates them by hand-written optimizations.
Ding et al.~\cite{glore} proposed GLORE, which mainly focuses on loop-invariant reductions. For example, it is able to identify the redundant reduction $\sum_k B(j,k)$ in $\sum_{i,j,k} A(i,j)B(j,k)$. In contrast, RACE primarily targets the loop-carried redundant expressions, e.g.  $\cos(\ulat (i,j))$ in iteration $(i,j)$ is identical to  $\cos(\ulat (i,j-1))$ in iteration $(i,j+1)$.

Redundancy elimination, as well as many other compiler optimizations, such as loop transformation, memory access optimization and vectorization, relies on the results of aliasing analysis~\cite{alias-analysis, interprocedual-alias}. 
Alias analysis disambiguates memory references, primarily serving for the legality verification or candidate selection. 
For example, existing work like ESR~\cite{Cooper+:pact08} selects the redundancy candidates using the results of data dependence analysis which implicitly incorporates alias analysis. 
Thus, alias analysis is orthogonal to these optimizations to some extent.
Furthermore, existing aliasing analysis techniques are powerful but not panaceas.
Other mechanisms such as the \texttt{restrict} keyword, \texttt{-fno-alias} compiler option, and \texttt{\#pragma ivdep} directive are provided to enable more aggressive optimizations.
Therefore, in this paper, RACE focuses on the detection technique of redundancies, instead of directly incorporating alias analysis.

\section{Overview}
\label{sec:overview}

\subsection{Scope and Assumptions}
\label{sec:target}

This paper focuses on perfectly nested loops without internal control flow.
RACE eliminates redundant computations in unmodified arrays within the loop.
Each candidate expression considered in the algorithm involves arithmetic operations on scalar or array variables.
RACE targets array references of the form $A[a_1i_{s_1}+b_1]\cdots[a_ni_{s_n}+b_n]$
(in the following text, array notations may appear in either C-style (e.g., $A[i][j]$)
or Fortran-style (e.g., $A(i,j)$), depending on the context),
where $s_k$ is a loop level number ranking from the outermost loop to the innermost,
$1\leqslant s_k\leqslant m$, $m$ is the nesting depth of the loop,
$i_{s_k}$ is a loop index, $a_k$ and $b_k$ are constants.
For all the array references in the above POP example in Figure~\ref{fig-pop1},
we have $s_1=2$ and $s_2=1$, $i_{s_1}=i$ and $i_{s_2}=j$.
This form covers most loop nests in real-world applications.
For the applications used in the evaluation of this work,
over 95\% of array accesses in loops conform to this pattern.
Scalars are treated as zero-dimensional array references.
Function calls such as $\sin$ and $\cos$ can be interpreted as binary operators $\odot$, where the function name is treated as a scalar variable.
For example, $\cos(ulat(1,j-1))$ can be viewed as a binary expression $\cos \odot ulat(1,j-1)$, where $\cos$ and $\ulat(1,j-1)$ are the left and right operands.

\subsection{RACE Framework}

\begin{figure}[!t]
  \centering
  \includegraphics[width=0.98\textwidth]{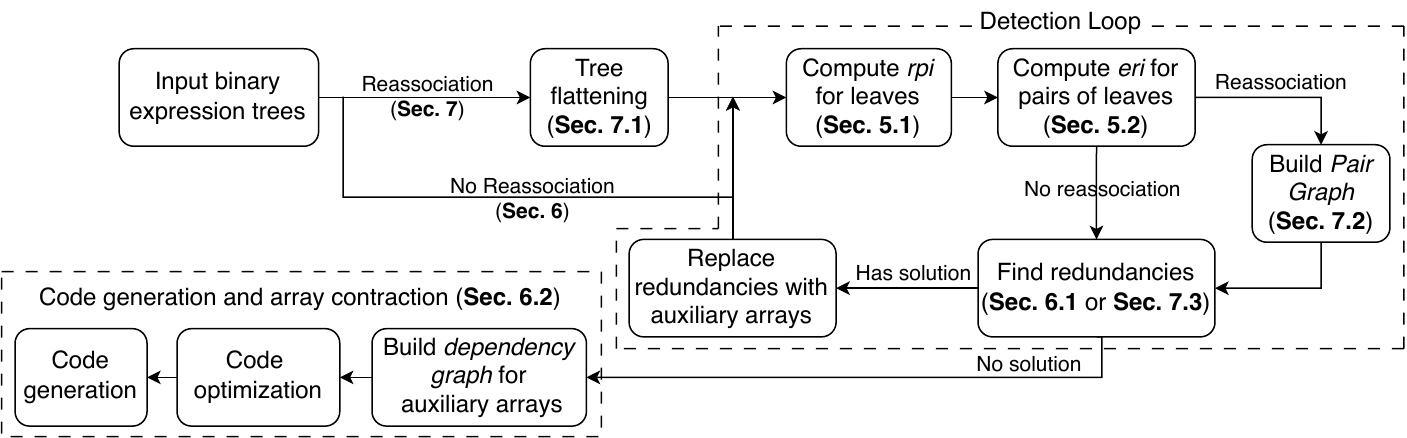}
  \caption{Framework of the RACE method. "No reassociation" and "reassociation" correspond to the detection procedures on binary tree (Section~\ref{sec:detectionbinary}) and n-ary tree (Section~\ref{sec:detectionnary}), respectively.}
  \label{fig:overview}
  \vspace{-0.2cm}
\end{figure}

Figure~\ref{fig:overview} illustrates the main components and overall workflow of RACE.
The subsequent sections present the details of the components and algorithms involved.
RACE takes binary expression trees inside a perfectly nested loop as input.
The detection algorithm for binary tree (Section~\ref{sec:detectionbinary}) follows the "no reassociation" path. It preserves the original evaluation order of expressions, thereby preserving floating-point results. The detection algorithm for n-ary tree (Section~\ref{sec:detectionnary}) follows the "reassociation" path, which restructures the input trees through tree flattening (Section~\ref{subsec:flattening}) before the detection loop, exposing additional redundancies.

RACE detects redundancy iteratively. In each iteration of the detection loop, it identifies redundant binary subexpressions, replaces them with references to auxiliary arrays, and then continues on the transformed trees so that larger redundant expressions can be exposed in later iterations. This process is driven by a two-level hash-based identification scheme: RACE first computes the reference pattern identifier $\rpi$ (Section~\ref{subsec:rpi}) for array references and scalar leaves, and then computes the expression redundancy identifier $\eri$ (Section~\ref{subsec:eri}) for binary subexpressions. The binary-tree algorithm enables linear-time detection, whereas the n-ary-tree algorithm uses the Pair Graph (Section~\ref{subsec:nary-algorithm}) to resolve conflicts among multiple candidate redundancies. After detection, RACE generates optimized code and then applies auxiliary-array optimization to reduce storage overhead and simplify the transformed program (Section~\ref{subsec:optimization}).

\section{Redundancy Identification}
\label{sec:identification}

This section describes the two-level redundancy identification scheme.
The first level defines a hash function for array references to identify data sharing.
Based on the first level, the second level defines a hash function for binary expressions, and redundancies can be efficiently discovered by grouping the expressions with the same hash value.

\subsection{Reference Pattern Identifier ($\rpi$)}
\label{subsec:rpi}

To recognize the redundancy among expressions, the data sharing among array references needs to be analyzed.
Existing data dependence techniques test whether two array references may access at least one common array element.
On the contrary, the redundancy elimination of array computations requires the reference of a sufficiently large set of common array elements.
This paper introduces array \emph{reference pattern}
and only considers the redundancy among array references with the same reference pattern. 

Intuitively, two array references share the same reference pattern only if they access the same infinite lattice in the iteration space, assuming that each loop index $i_k$ ranges over $\mathbb{Z}$.
For example, $A[i][j]$ and $A[i+1][j-1]$ both touch all the lattice points in $\mathbb{Z}^2$, while the reference lattices of $A[2i]$ and $A[2i+1]$ are disjoint, and $A[2i]$ and $A[3i]$ access partially overlapping lattice points.

Formally, an array reference of the form $A[a_1i_{s_1}+b_1]\cdots[a_ni_{s_n}+b_n]$ accesses the infinite integer lattice
$\mathcal{L}(\mathbf{B},\mathbf{b}) = \{\mathbf{B} \mathbf{x} + \mathbf{b} \mid \mathbf{x}=(i_1,\cdots, i_m)\in \mathbb{Z}^m\},$
where $m$ is the nesting depth of the loop, $n$ is the dimension of the array $A$, $\mathbf{B}$ is an $n\times m$ basis matrix with $\mathbf{B}_{k,s_k} = a_k$ and all other entries equal to 0, and $\mathbf{b}$ is an $n$-dimensional offset vector whose $k$-th component is $b_k$.
According to lattice theory, two array references access the same infinite lattice in $\mathbb{Z}^n$ if and only if there exists a unimodular matrix $\mathbf{U}$ such that
$\mathbf{B} \mathbf{U} = \mathbf{B}' \text{ and } \mathbf{b} - \mathbf{b}' \in \mathcal{L}(\mathbf{B},\mathbf{0}),$
where $\mathbf{B}, \mathbf{b}$ and $\mathbf{B}', \mathbf{b}'$ denote their basis matrices and offset vectors.

A unimodular matrix can be represented as a product of elementary column operations, including: 1. column exchange ($A[i][j]$ to $A[j][i]$); 2. column multiplied by $-1$ ($A[i][j]$ to $A[-i][j]$); and 3. adding an integer multiple of one column to another ($A[i][j]$ to $A[i+2j][j]$).
However, the reference patterns identified in this paper do not cover general unimodularly related lattices, for two practical reasons. 
First, the third operation does not appear since we have restricted the form of array subscripts previously, i.e. each row of $\mathbf{B}$ has only one nonzero entry. We also rarely detect redundancies that involve the first two operations in the evaluated applications.
Second, although the unimodular condition is concise for infinite lattices, it is difficult in practice to decide whether two unimodularly related finite lattices overlap on a sufficiently large region~\cite{Clauss96}, especially when loop bounds are unknown at compile time.
Therefore, to simplify detection, we define the reference pattern as follows. Two array references have the same reference pattern \wzx{if and only if} their lattices $\mathcal{L}(\mathbf{B},\mathbf{b})$ and $\mathcal{L}(\mathbf{B}', \mathbf{b}')$ satisfy:
1. $\mathbf{B} = \mathbf{B}'$, and 2. $\mathbf{b} - \mathbf{b}' \in \mathcal{L}(\mathbf{B},\mathbf{0})$ (i.e., $\mathbf{B} \mathbf{x} = \mathbf{b} - \mathbf{b}'$ has an integral solution).

Checking these two conditions appears to require pairwise comparisons between array references. However, the information required by the two conditions can be encoded locally for each individual array reference.
We introduce a hash function, called \emph{reference pattern identifier} ($\rpi$), to perform this encoding. If two array references have the same $\rpi$ value, they satisfy the above two conditions and thus have the same reference pattern.
The first condition ($\mathbf{B} = \mathbf{B}'$) is straightforward.
It follows that if two references share the same reference pattern, their index list (dimension axis) $\langle i_{s_1},\cdots,i_{s_n} \rangle$ and the coefficient list $\langle a_1,\cdots,a_n \rangle$ must be the same, respectively. Thus, $\rpi$ must incorporate these two lists.

The second condition ($\mathbf{b} - \mathbf{b}' \in \mathcal{L}(\mathbf{B},\mathbf{0})$) needs more elaboration.
Consider two potential references $A[a_1i_{s_1}+b_1]\cdots[a_ni_{s_n}+b_n]$ and $A[a_1i_{s_1}+b_1']\cdots[a_ni_{s_n}+b_n']$.
If a loop index $i_{s_k}=i$ appears only once in the subscripts, 
$a_ki = b_k-b_k'$ has an integral solution if and only if
$b_k \bmod a_k = b_k' \bmod a_k$.
Thus, $\rpi$ incorporates
$b_k \bmod a_k$.
For example, $A[2i]$ and $A[2i+2]$ share the same $\rpi$. 
If $i$ appears in multiple subscripts, the successive deltas (differences) must also match. 
Assuming $i_{s_k}=i_{s_j}=i$,
the equations 
$a_k i = b_k - b_k'$ and 
$a_j i = b_j - b_j'$
must share the same integer solution.
Besides the condition $b_k \bmod a_k = b_k' \bmod a_k$,
it is also required that $b_k/a_k - b_j/a_j = b_k'/a_k - b_j'/a_j$. Therefore $\rpi$ should also incorporate $b_k/a_k - b_j/a_j$.
For example, $A[2i+1][3i+2]$ and $A[2i+3][3i+5]$ have the same delta value $2/3-1/2=5/3-3/2=1/6$ for $i$.
Thus they access the same infinite lattice.

\begin{algorithm}[h]
    \caption{Calculating information for $\rpi$}
    \label{alg:rpi}
    \SetKwInOut{Input}{input}
    \SetKwInOut{Output}{output}
    \Input{An array reference $x=A[a_1 i_{s_1} + b_1]\cdots[a_n i_{s_n} + b_n]$}
    \Output{$\mathit{indexList},\mathit{indexCoef},\mathit{indexDelta},\mathit{firstIndexOffset}$ of $x$}
    \tcp{$x.\mathit{indexList}$ and $x.\mathit{indexCoef}$ are lists}
    \tcp{$x.\mathit{firstIndexOffset}$ is an array of rationals}
    \tcp{$x.\mathit{indexDelta}$ is an array of lists of rationals}

    $x.\mathit{firstIndexOffset}[1\cdots m] = \infty$; \tcp{$m$ is the nesting depth of the loop} \
    \For{$k \gets 1$ \KwTo $n$}{
        \uIf{$a_k \neq 0$}{
            $x.\mathit{indexList}.$append$(s_k)$\;
            $x.\mathit{indexCoef}.$append$(a_k)$\;
            \uIf{$x.\mathit{firstIndexOffset}[s_k] = \infty$}{
                $x.\mathit{firstIndexOffset}[s_k] = b_k/a_k$\;
                $x.\mathit{indexDelta}[s_k]$.append($b_k \bmod a_k$)\;
            } \Else{
                $x.\mathit{indexDelta}[s_k]$.append($b_k/a_k - x.\mathit{firstIndexOffset}[s_k]$) \;
            }
        } \Else{
            $x.\mathit{indexList}.$append$(0)$\; 
            $x.\mathit{indexCoef}$.append($b_k$) \;
        }
    }
\end{algorithm}

Algorithm~\ref{alg:rpi} shows the pseudo-code that extracts the information for $\rpi$. The descriptions of the variables used are given below:
$\mathit{indexList}$ and $\mathit{indexCoef}$ store the loop index variables and their corresponding coefficients in the order they appear.
$\mathit{firstIndexOffset}$ stores the normalized offset of each loop index variable when it first appears. It is used to compute $\mathit{indexDelta}$, which contains the successive deltas of the same loop index variable. In calculating $\mathit{indexDelta}$, it keeps the rational results in irreducible fraction form.
If a loop index is missing, i.e. $a_k = 0$, it appends a virtual loop level $0$ and the corresponding constant $b_k$ must be added to the coefficient list. For example, $A[i][1]$ and $A[i][2]$ do not share common elements, and the coefficients $1$ and $2$ are the only identifiers in this dimension. 
Let $x.\mathit{name}$ denote the name of the array reference $x$.
Putting it all together, $\rpi$ is formulated as follows:

$$\rpi(x) = \hash(x.\mathit{name},x.\mathit{indexList},x.\mathit{indexCoef},x.\mathit{indexDelta}).$$

\subsection{Expression Redundancy Identifier ($\eri$)}
\label{subsec:eri}

\wzx{With $\rpi$ of each operand, we define the expression redundancy identifier ($\eri$) to recognize expressions with redundant computations.
$\eri$ also uses the local-encoding idea to avoid pairwise comparisons.
Expressions sharing the same $\eri$ value are considered redundant.}
Intuitively, $\rpi$ identifies potential reuse at the array-reference level, while $\eri$ further distinguishes whether two expressions align in their operand index shifts.
$\eri$ cannot be as simple as the classical value number definition like $vn(e)=\hash(vn(x), \oplus, vn(y))$.
The expressions may still not represent redundant computations, even if the corresponding operands in two expressions share the same $\rpi$ values. 
For example, the two expressions $e=A[i]+B[i]$ and $e'=A[i+1]+B[i+2]$ are not redundant, although $\rpi(A[i])=\rpi(A[i+1])$ and $\rpi(B[i])=\rpi(B[i+2])$.
Therefore, we introduce the expression-level delta value $\mathit{exprDelta}$.
Algorithm \ref{alg:eri} presents the pseudo-code for calculating 
this information for $\eri$. 
$\mathit{exprDelta}$ stores the delta between the two operands of the expression for each loop index variable. For example, supposing the loop nest order is $\langle i,j,k \rangle$, for $e=A[i][2j+1]+B[2i+3][k]$, we have $A[i][2j+1].\mathit{firstIndexOffset}=\langle 0,1/2,\infty \rangle, B[2i+3][k].\mathit{firstIndexOffset}=\langle 3/2,\infty,0 \rangle$, thus $e.\mathit{exprDelta}=\langle -3/2,\infty,\infty \rangle$.

\begin{algorithm}[!htb]
    \caption{Calculating information for $\eri$}
    \label{alg:eri}
    \SetKwInOut{Input}{input}
    \SetKwInOut{Output}{output}
    \Input{An expression $e=x\oplus y$}
    \Output{$\mathit{exprDelta}$ of $e$}
    \tcp{$e.\mathit{exprDelta}$ is an array of rationals}
    $e.\mathit{exprDelta}[1\cdots m] = \infty$ \;
    \For{$s_k \in \{x.\mathit{indexList}\} \cap \{y.\mathit{indexList}\} - \{0\}$}{
            $e.\mathit{exprDelta}[s_k] = x.\mathit{firstIndexOffset}[s_k] - y.\mathit{firstIndexOffset}[s_k]$\;
    }
\end{algorithm}

By incorporating $\mathit{exprDelta}$ into the hash, $\eri$ captures the expression-level alignment information and yields identical values only when both the operand reference patterns and their relative index offsets are the same, which ensures correctness for expression-level redundancy detection.
The formula of $\eri$ for $e=x\oplus y$ is as follows:
$$\eri(e) = \hash(\rpi(x),\oplus,\rpi(y),e.\mathit{exprDelta}).$$

For commutative operators such as addition and multiplication, the two operands need to be sorted to recognize the redundancy.
For example, $A[i]+B[i]$ and $B[i+1]+A[i+1]$
are considered as redundant computations.
The two operands can be sorted by their $\mathit{name}$
$A$ and $B$ lexicographically.
If the two operands access the same array,
they are sorted according to their other $\rpi$ information:
$\mathit{indexList}$, $\mathit{indexCoef}$ and $\mathit{indexDelta}$.
Thus, the redundancies between 
$A[i]+A[2i]$ and $A[2i+2]+A[i+1]$,
$A[i]+A[i+1]$ and $A[i+2]+A[i+1]$
can be found.
Without loss of generality, we assume that the operands in $e=x\oplus y$ are already in order.

\section{Redundancy Elimination on Binary Trees}
\label{sec:detectionbinary}

\begin{figure}[!t]
  \centering
  \includegraphics[width=0.98\textwidth]{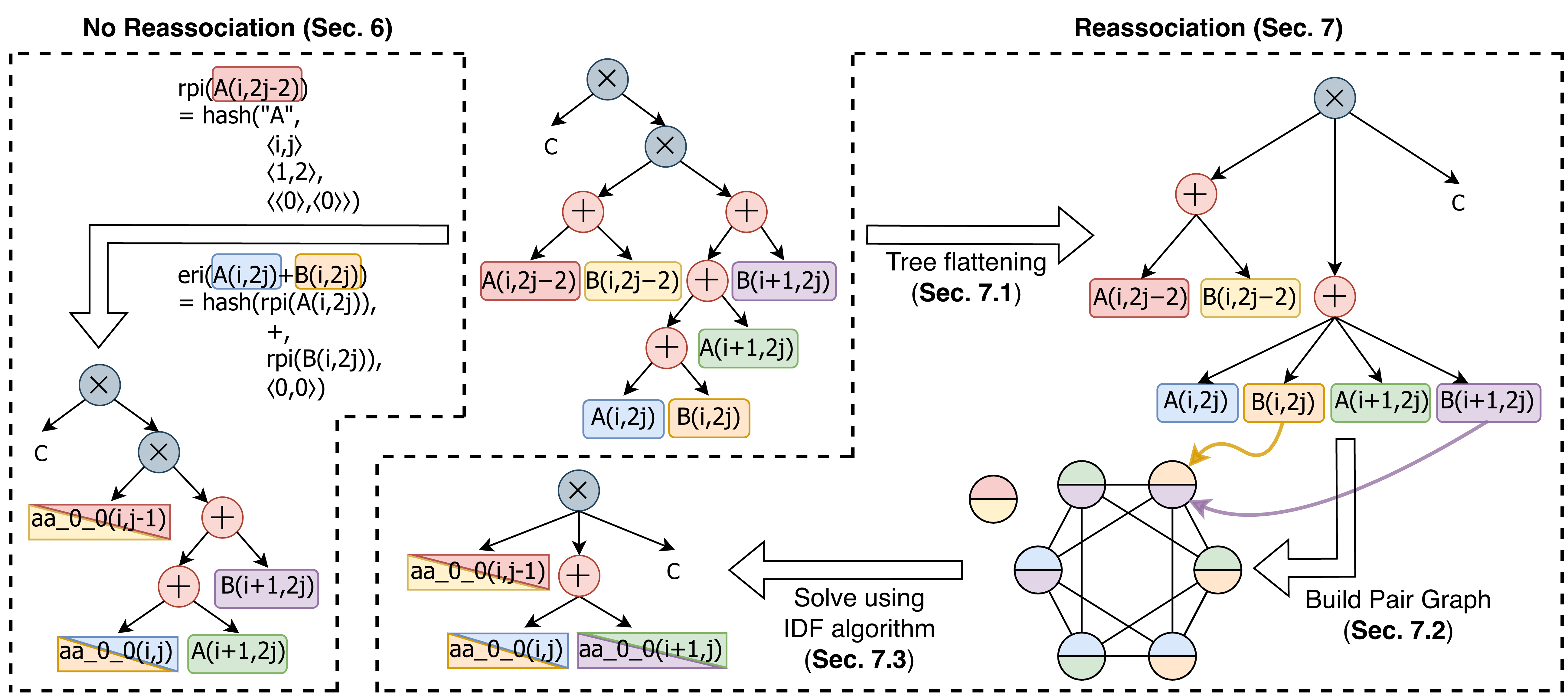}
  \caption{Example of redundancy detection in RACE. The middle shows the input expression tree. The left illustrates the binary-tree detection algorithm (Section~\ref{sec:detectionbinary}), which detects redundancies while preserving the original evaluation order. The right illustrates the n-ary-tree detection algorithm (Section~\ref{sec:detectionnary}), where tree flattening, Pair Graph construction, and the IDF algorithm enable the detection of additional redundancies. The detected redundancies are replaced with auxiliary array references, yielding different transformed expression trees.}
  \label{fig:example}
\end{figure}

\subsection{Detection Algorithm}
\label{subsec:detectionbinary}

This section presents the redundancy detection algorithm for binary expression trees. Two reasons motivate this choice of intermediate representation (IR).
The first is to preserve result consistency, since binary trees fix the computation order and thus maintain identical numerical results under floating-point arithmetic. The second reason is that array redundancies naturally form a hierarchical dependence structure, which cannot be captured by a linear IR (e.g., three-address code). For example, in $a=B[\dots]\oplus C[\dots]$, the $\rpi$ values of the operands and the $\eri$ of $a$ can be computed. 
When a subsequent statement $e=a \oplus D[\dots]$ appears, $\rpi(a)$ is not informative for array-aware redundancy detection, because scalar $a$ no longer carries the reference information of $B[\cdots]$ and $C[\cdots]$ until it is replaced by an auxiliary array $F[\dots]$. This contrasts with the classical value numbering method, which works on scalar expressions in three-address code IR. It can detect all redundancies in a single sequential traversal of the IR, since for a scalar assignment $a=b\oplus c$, the value number of $a$ can be directly computed from those of $b$ and $c$, and then used in subsequent redundancy checks.

The detection algorithm follows the "no reassociation" path in Figure~\ref{fig:overview}, with an example shown in the left panel of Figure~\ref{fig:example}. It enters the detection loop without modifying the structure of the input tree. First, it computes the $\rpi$ values for all leaf nodes. Second, it assigns $\eri$ values to operator nodes whose two children are leaves.
Expressions with the same $\eri$ value are then treated as redundant and replaced by new auxiliary array references.
These new references must maintain consistent subscripts.
To achieve this, one representative expression is selected for each auxiliary array.
For example, $A(i,2j)+B(i,2j)$ and $A(i,2j-2)+B(i,2j-2)$ share the same $\eri$,
and an auxiliary array $aa\_0\_0$ is allocated for them.
A representative expression is assigned to $aa\_0\_0$, e.g., $aa\_0\_0(i,j) = A(i,2j)+B(i,2j)$.
The other redundant expressions are replaced by references to this array with corresponding index offsets, e.g., $A(i,2j-2)+B(i,2j-2)$ will be replaced by $aa\_0\_0(i,j-1)$.
Finally, the newly introduced auxiliary array references are treated as new leaves.
The detection loop continues until no new redundancies are detected.

The algorithm detects and organizes redundancies hierarchically, as illustrated by the POP example in Section~\ref{sec-motivation}.
In each iteration of the detection loop (except the first), new redundant expressions are discovered, whose operands include auxiliary arrays generated in the previous iteration.
This iterative process is valid because redundancies are composable.
For instance, if $(x \oplus y) \oplus z$ is redundant with $(x' \oplus y') \oplus z'$, it can be identified in the next iteration after $x \oplus y$ is first recognized as redundant and replaced with an auxiliary array.

The algorithm runs in linear time complexity:
the redundancies can be identified in a single bottom-up traversal of the binary expression tree.
The $\rpi$ values of the leaves and the $\eri$ values of the expressions are stored directly within the tree nodes.
The proposed two-level hash scheme in Section \ref{sec:identification} simplifies the redundancy detection procedure, so that the expressions sharing the same $\eri$ can be identified as redundant and replaced by auxiliary arrays references as new tree leaves. In the next iteration, only the newly introduced leaves and expressions need to be processed.

\subsection{Code Generation and Array Contraction}
\label{subsec:optimization}

\begin{figure}[!h]
  \centering
  \begin{tabular}{c c}
    \begin{minipage}{0.13\textwidth}
    \end{minipage}
    \begin{minipage}{0.6\textwidth}
      \tikzstyle{vertex} = [rectangle, inner sep=-4pt, fill=none, minimum size=14pt, draw=none,font=\ttfamily]
      \begin{tikzpicture}[scale=1]
        \node[vertex] (a12) at (0,0) {\small aa\_1\_2};
        \node[vertex] (a21) at ($(a12)+(-2.4,1.1)$) {\small aa\_2\_1};
        \node[vertex] (a20) at ($(a21)+(-1.5,0)$) {\small aa\_2\_0};

        \node[vertex, circle, draw=black] (tx) at ($(a20)+(0,1.1)$) {\small tx};
        \node[vertex, circle, draw=black] (ty) at ($(a21)+(0,1.1)$) {\small ty};
        \node[vertex, circle, draw=black] (tz) at ($(a12)+(0,2.2)$) {\small tz};
        
        \node[vertex] (a03) at ($(a12)+(0,-1.1)$) {\small aa\_0\_3};
        \node[vertex] (a11) at ($(a21)+(0,-1.1)$) {\small aa\_1\_1};
        \node[vertex] (a10) at ($(a20)+(0,-1.1)$) {\small aa\_1\_0};
      
        \node[vertex] (a02) at ($(a11)+(0.75,-1.1)$) {\small aa\_0\_2};
        \node[vertex] (a01) at ($(a11)+(-0.75,-1.1)$) {\small aa\_0\_1};
        \node[vertex] (a00) at ($(a10)+(-0.75,-1.1)$) {\small aa\_0\_0};

        \draw [cyan, fill=cyan, fill opacity=0.0] plot [smooth cycle] coordinates {($(a10)+(-0.3,0.2)$) ($(a00)+(-0.35,-0.3)$) 
          ($(a03)+(0.3,-0.3)$) ($(a03)+(0.3,0.3)$) ($(a11)+(0.3,0.25)$)};

        \draw [orange, fill=orange, fill opacity=0.0] plot [smooth cycle] coordinates {($(a20)+(-0.35,-0.3)$) ($(a20)+(-0.25,0.3)$) ($(a21)+(0.4,0.3)$) 
          ($(a12)+(0.4,0.3)$) ($(a12)+(0.45,-0.2)$) ($(a12)+(-0.35,-0.2)$) ($(a11)+(0.5,0.65)$) };

        \draw [teal, fill=teal, fill opacity=0.0] plot [smooth cycle] coordinates {($(tx)+(0.05,-0.35)$) ($(tx)+(0.05,0.35)$)
          ($(tz)+(-0.05,0.35)$) ($(tz)+(-0.05,-0.35)$) };

        \node[vertex] (btx) at ($(tx)+(-1.1,0)$) {\scriptsize $[2,nx][2,ny]$};
        
        \node[vertex] (dtxl) at ($(tx)!0.5!(a20)-(0.7,0)$) {\footnotesize $(i-1,j)$};
        \node[vertex] (dtxr) at ($(tx)!0.5!(a20)+(0.5,0)$) {\footnotesize $(i,j)$};
        \node[vertex] (b20) at ($(a20)+(-1.3,0)$) {\scriptsize $[1,nx][2,ny]$};
        
        \node[vertex] (d20l) at ($(a20)!0.5!(a10)-(0.7,0)$) {\footnotesize $(i,j-1)$};
        \node[vertex] (d20r) at ($(a20)!0.5!(a10)+(0.5,0)$) {\footnotesize $(i,j)$};
        \node[vertex] (b10) at ($(a10)+(-1.3,0)$) {\scriptsize $[1,nx][1,ny]$};
        
        \node[vertex] (d10l) at ($(a10)!0.5!(a00)-(0.25,0)$) {\footnotesize $(i,j)$};
        \node[vertex] (b00) at ($(a00)+(-1.3,0)$) {\scriptsize $[1,nx][1,ny]$};

        \newcommand{\doublearrow}[2]{
          \path[-stealth] (#1) edge[bend left=15] (#2);
          \path[-stealth] (#1) edge[bend right=15] (#2);
        }
      
        \foreach \x/\y in {a12/a03, a21/a11, a20/a10, tx/a20, ty/a21}{
          \doublearrow{\x}{\y};
        }

        \path[-stealth] (tz) edge[bend right=10] (a12);
        \path[-stealth] (tz) edge[bend left=10] (a12);
      
        \path[-stealth] (a11) edge[bend right=10] (a02);
        \path[-stealth] (a11) edge[bend left=10] (a01);
        \path[-stealth] (a10) edge[bend right=10] (a01);
        \path[-stealth] (a10) edge[bend left=10] (a00);

      \end{tikzpicture}
    \end{minipage}
    &
    \begin{minipage}{0.2\textwidth}
      \tikzstyle{vertex} = [rectangle, rounded corners, fill=none,minimum size=0pt, inner sep=-4pt, draw=black,font=\ttfamily]
      \begin{tikzpicture}[scale=1.2]

        \node[vertex, draw=none] (doj) at (0,0) {\small \texttt{DO j=2,ny}};
        \node[vertex, draw=none] (doi) at ($(doj)+(0.2,-0.35)$) {\small \texttt{DO i=2,nx}};
        \node[vertex, draw=none] (arrow0f) at ($(doj)+(-0.7,+0.23)$) {};
        \node[vertex, draw=none] (arrow0t) at ($(doj)+(-0.3,+0.23)$) {};
        \node[vertex, draw=none] (arrow1f) at ($(doj)+(-0.7,-0.17)$) {};
        \node[vertex, draw=none] (arrow1t) at ($(doj)+(-0.3,-0.17)$) {};
        \node[vertex, draw=none] (arrow2f) at ($(doi)+(-0.7,-0.17)$) {};
        \node[vertex, draw=none] (arrow2t) at ($(doi)+(-0.3,-0.17)$) {};

        \path[-stealth] (arrow0f) edge[bend left=0] (arrow0t);
        \path[-stealth] (arrow1f) edge[bend left=0] (arrow1t);
        \path[-stealth] (arrow2f) edge[bend left=0] (arrow2t);

        \node[vertex, draw=none] (note0) at ($(arrow0f)+(-0.1,0)$) {0};
        \node[vertex, draw=none] (note1) at ($(arrow1f)+(-0.1,0)$) {1};
        \node[vertex, draw=none] (note2) at ($(arrow2f)+(-0.1,0)$) {2};

      \end{tikzpicture}
    \end{minipage}
     \\
  \end{tabular}
  \caption{Dependency graph of the auxiliary arrays in POP benchmark.}
  \label{fig:dependency-graph}
\end{figure}

\emph{Code Generation.}
The expression trees are optimized with auxiliary arrays after the execution of RACE. Since the previous auxiliary arrays may be used by the newer ones during the iterative process, the dependencies between auxiliary arrays are represented by a directed acyclic graph (the \emph{dependency graph}) with a natural topological order. Figure~\ref{fig:dependency-graph} illustrates the dependency graph of auxiliary arrays in POP benchmark.

To generate the straightforward code (left of Figure~\ref{fig-pop2}), we need (i) a topological order and (ii) loop ranges. The top nodes $\mathit{tx}, \mathit{ty}, \mathit{tz}$ inherit the original loop range $[2,\mathit{nx}][2,\mathit{ny}]$. Using the representative expressions, each parent propagates its reference offsets to determine the child's range. Ranges are computed in topological order. For example, $aa\_2\_0$ is referenced at $(i,j)$ and $(i-1,j)$ by $\mathit{tx(i,j)}$, so its range becomes $[1,\mathit{nx}][2,\mathit{ny}]$. Figure~\ref{fig:dependency-graph} annotates nodes with their ranges and reference offsets.

Nodes with identical ranges are grouped into a \emph{range circle} (colored circle in Figure~\ref{fig:dependency-graph}); arrays in the same circle are initialized together by one loop. 
We then build a coarse graph over range circles, where each range circle corresponds to a node in the coarse graph, and generate loops in its topological order; inside a circle we still respect the fine-grained dependency order. This yields the code on the left of Figure~\ref{fig-pop2}.

\emph{Array Contraction.}
This simple code allocates a large space for the auxiliary array. However, it is possible to decrease the auxiliary dimensionality and reuse a small array space. Several contraction rules are given as follows:
First, if an auxiliary array has a single parent, inline its representative expression instead of storing it. In Figure~\ref{fig-pop2}, $aa\_0\_0$ and $aa\_0\_2$ are replaced by $\cos(\ulon(i,j))$ and $\sin(\ulon(i,j))$.
Second, if an auxiliary array and all its parents are in the same range circle, all reference offsets must be zero, so it can be a scalar. Thus $aa\_0\_1$ becomes $aa\_0\_1\_i\_j$ (Figure~\ref{fig-pop2}, right).
Third, the dimension of auxiliary arrays can be eliminated or compressed by inserting precompute loops at inner insertion positions (numbered in Figure~\ref{fig:dependency-graph}, right). Treat each circle as a unit and, following the topological order, check positions from 0 inward. If a circle's range along dimension $j$ matches that of the original loop, move it to position 1 and eliminate dimension $j$ from all arrays inside it.
The check can repeat inward until it fails. The orange circle satisfies this, so $aa\_2\_0, aa\_2\_1, aa\_1\_2$ lose dimension $j$ (Figure~\ref{fig-pop2}, right). Note that the innermost dimension $i$ is retained to facilitate vectorization.
Even if the above dimension range check fails, the dimension can still be compressed. For example, the blue circle's $j$ range $[1,\mathit{ny}]$ differs from the original $[2,\mathit{ny}]$ due to $j-1$ references offset. Rather than placing the loop at position 0, we split it into a boundary prefetch $[1,1]$ at position 0 and a main part $[2,\mathit{ny}]$ at position 1, allocate only two $j$-slices ($j_0,j_1$) and swap them each iteration, which is a double buffer of size $\mathit{nx}\times2$. $aa\_0\_3$, $aa\_1\_0$, and $aa\_1\_1$ illustrate this situation (Figure~\ref{fig-pop2}, right).

\subsection{Redundancy Analysis}
\label{subsec:redundancy-analysis}

After identifying the expressions with the same $\eri$ value, these redundant expressions can be replaced with precomputed auxiliary arrays. The range information for each auxiliary array is computed in the dependency graph. The evaluation of the profit gained from the replacement, i.e., the number of arithmetic operations eliminated, is discussed in this section.

Suppose that the auxiliary arrays $aa_1, \cdots, aa_m$ appear in the transformed expression trees while the remaining positions are the original tree nodes. Let $cnt(aa_k)$ denote the number of occurrences of $aa_k$ in the expression trees. For loop ranges, only the difference between the lower and upper bound is considered, which is denoted by $r(i_t)$ for the original loop range at index $i_t$. Let $r(i_t,aa_k)$ denote the precompute loop range of auxiliary array $aa_k$ at index $i_t$. 

The reduction in arithmetic operations is only related to the modified parts of the expression trees. The first step computes the original cost of the replaced parts.
This is done by recursively expanding each auxiliary array's representative expression: if $aa_k$ contains another auxiliary array $aa_p$, it is replaced by $aa_p$'s representative expression until no auxiliary arrays remain.
Let the number of arithmetic operations in the expanded representative expression of $aa_k$ be denoted by $ops(aa_k)$. The original computational cost of the positions replaced by the auxiliary arrays in the expression tree can be computed by: 
\[ori= \prod_{t=1}^n r(i_t) \sum_{k=1}^m ops(aa_k) \cdot cnt(aa_k). \]

The generated code replaces the $ori$ operations with auxiliary array reads. However, generating the auxiliary arrays requires precomputation loops. Therefore, the computational cost of this precomputation must also be considered to assess the benefit. Since our $\eri$ method targets binary expressions, each auxiliary array's precompute expression is a binary expression. Thus, the precomputation cost of the auxiliary arrays is:
\[aft=\sum_{k=1}^m \prod_{t=1}^n r(i_t,aa_k). \]
Hence, the profit gained from the detection algorithm on binary tree is $ori-aft$. 

\section{Redundancy Elimination on n-ary Trees}
\label{sec:detectionnary}

This section enhances redundancy detection for processing n-ary expression trees. 
The "reassociation" path in Figure~\ref{fig:overview} illustrates the redundancy detection on n-ary trees.
The input binary expression trees are first transformed into n-ary trees before entering the detection loop.
The motivation and strategies for this transformation are discussed in Section~\ref{subsec:flattening}.
In each iteration of the detection algorithm, the following steps are performed: compute $\rpi$ in the same way as before; for every operator node in the n-ary expression tree, compute $\eri$ for every pair of its leaf children (operands); build the Pair Graph (defined in Section~\ref{subsec:nary-algorithm}) with $\eri$ as the color of each node; identify redundancies on the graph, with the objective function formulated
in Section~\ref{subsec:nary-algorithm}, and a heuristic method proposed in Section~\ref{subsec:dimension-first}; finally, replace multiple redundant binary subexpressions with auxiliary arrays. The iteration continues until no further redundancy is found on the graph.

\subsection{Tree Flattening}
\label{subsec:flattening}

\begin{figure}[h]
\centering
\begin{tabular}{c|c|c}
  \begin{minipage}[t]{0.38\textwidth}
    \centering
    \begin{lstlisting}[language=Fortran, basicstyle=\tiny\ttfamily]
DO j=2,n-1
DO k=2,n-1
DO i=2,n-1
 U(i,k,j)=U(i,k,j)
     +w0*R(i,k,j)
     +w1*(R(i-1,k,j)+R(i+1,k,j)
         +R(i,k-1,j)+R(i,k+1,j)
         +R(i,k,j-1)+R(i,k,j+1))
     +w2*(R(i-1,k-1,j)+R(i+1,k-1,j)
         +R(i-1,k+1,j)+R(i+1,k+1,j)
         +R(i,k-1,j-1)+R(i,k+1,j-1)
         +R(i,k-1,j+1)+R(i,k+1,j+1)
         +R(i-1,k,j-1)+R(i-1,k,j+1)
         +R(i+1,k,j-1)+R(i+1,k,j+1))
     +w3*(R(i-1,k-1,j-1)+R(i+1,k-1,j-1)
         +R(i-1,k+1,j-1)+R(i+1,k+1,j-1)
         +R(i-1,k-1,j+1)+R(i+1,k-1,j+1)
         +R(i-1,k+1,j+1)+R(i+1,k+1,j+1))
ENDDO 
ENDDO 
ENDDO
    \end{lstlisting}
  \end{minipage}
  &
  \begin{minipage}[t]{0.3\textwidth}
    \centering
    \begin{lstlisting}[language=Fortran, basicstyle=\tiny\ttfamily]
DO j=2,n-1
DO k=2,n-1
DO i=1,n
 aa_1_0(i)=R(i,k-1,j)
          +R(i,k+1,j)
          +R(i,k,j-1)
          +R(i,k,j+1)
 aa_1_1(i)=R(i,k-1,j-1)
          +R(i,k+1,j-1)
          +R(i,k-1,j+1)
          +R(i,k+1,j+1)
ENDDO
DO i=2,n-1
 U(i,k,j)=U(i,k,j) 
         +w0*R(i,k,j) 
         +w1*(R(i-1,k,j)
             +R(i+1,k,j)
             +aa_1_0(i))
         +w2*(aa_1_0(i-1)
             +aa_1_0(i+1)
             +aa_1_1(i)) 
         +w3*(aa_1_1(i-1)
             +aa_1_1(i+1))
ENDDO 
ENDDO 
ENDDO
    \end{lstlisting}
  \end{minipage}
  &
  \begin{minipage}[t]{0.3\textwidth}
    \centering
    \begin{lstlisting}[language=Fortran, basicstyle=\tiny\ttfamily]
DO j=2,n-1
DO k=2,n-1
DO i=1,n
 aa_1_0(i)=R(i,k-1,j)
          +R(i,k+1,j)
          +R(i,k,j-1)
          +R(i,k,j+1)
 aa_1_1(i)=R(i,k-1,j-1)
          +R(i,k+1,j-1)
          +R(i,k-1,j+1)
          +R(i,k+1,j+1)
 aa_3_0(i)=w1*R(i,k,j)
          +w2*aa_1_0(i) 
          +w3*aa_1_1(i)
ENDDO 
DO i=2,n-1
 U(i,k,j)=U(i,k,j)
         +w0*R(i,k,j) 
         +w1*aa_1_0(i)
         +w2*aa_1_1(i)
         +aa_3_0(i-1)
         +aa_3_0(i+1)     
ENDDO 
ENDDO 
ENDDO
    \end{lstlisting}
  \end{minipage}
\\
\end{tabular}
\caption{An example from the mgrid benchmark.}
\label{fig:mgrid}
\end{figure}

Detecting redundancies in array or scalar computations may require \emph{reassociation}, i.e., reorganizing arithmetic calculations using commutative, associative, or distributive laws.
Consider $x+y+z$ and $x+z$.
Compilers may evaluate the former as $((x+y)+z)$ (left-associative), but the redundancy with $x+z$ cannot be exposed unless reassociation is permitted.
Note that the commutative law here differs from that used in $\eri$ calculation, where expressions have only two operands.

However, reassociation is generally illegal when result consistency must be preserved, even for integers.
For instance, with large positive $x,z$ but small negative $y$, both $((x+y)+z)$ and $(x+(y+z))$ may yield the correct value, whereas $((x+z)+y)$ may overflow.
Nevertheless, compilers usually provide options enabling such transformations for aggressive optimizations.

Therefore, the implementation of RACE provides several options to control the reassociation scheme.
We describe them using four aggressive levels. First, when consistency is required, reassociation is disabled and the linear redundancy detection algorithm is used. Second, parentheses in source code are respected, and reorganization can only be applied inside each pair of parentheses. Third, parentheses are removed when the inside operators are consistent with the one outside the parentheses. Finally, the distributive law is applied to remove more parentheses.

The first three levels are relatively straightforward.
Figure~\ref{fig:example} illustrates an example of flattening an expression tree using the third level, where the nodes $A(i,2j)$, $B(i,2j)$, $A(i+1,2j)$, and $B(i+1,2j)$ share the same parent node. The redundant expression $A(i+1,2j)+B(i+1,2j)$ is exposed, which can not be identified in the original input tree.
However, the distributive law should be applied cautiously as it may destroy existing redundancies.
For example, in $e=(A(i)+B(i))*(C(i)+D(i))+(A(i+1)+B(i+1))*(C(i+2)+D(i+2))$, two redundancies exist:
$\eri(A(i)+B(i))=\eri(A(i+1)+B(i+1))$ and $\eri(C(i)+D(i))=\eri(C(i+2)+D(i+2))$.
However, if the  distributive law is applied, both of the redundancies disappear
since none of the binary multiplication expressions has the same $\eri$.
Moreover, the distributive law may incur more computations.
For example, there are 2 multiplications and 5 additions in $e$, but 8 multiplications and 7 additions after removing the parentheses.
The detailed discussion is beyond the scope of this paper.

In practice, the distributive law is applied only when multiplying by a constant or loop-invariant scalar, which may yield additional redundancies.
Figure~\ref{fig:mgrid} shows a 3D stencil loop in the SPEC2000 mgrid benchmark.
The middle part shows the result at the second aggressive level, where auxiliary arrays $aa\_1\_0$ and $aa\_1\_1$ are found in the second iteration based on $aa\_0\_*$.
Using the fourth strategy reveals new redundancies, producing a new auxiliary array $aa\_3\_0$ in the fourth iteration.

Another set of options can specify the reassociation 
for subtraction and division operations.
For example, $x-y-z$ yields binary subexpressions including $x-y$, $x-z$, and $-y-z$.
To identify redundancy between $y+z$ and $-y-z$, subtractions are rewritten as addition with negated operands (e.g., $-y-z \rightarrow (-y)+(-z)$), and the first operand in $\eri$ computation is standardized to "$+$".
The division operation can be handled similarly.

\subsection{Problem Formulation and Transformation}
\label{subsec:nary-algorithm}

The algorithm on n-ary expression trees follows the iterative scheme as on binary trees.
However, redundancy identification on n-ary trees is more complex than on binary trees, where $\eri$ values are computed directly for each binary expression and redundancies are detected when the $\eri$ values are identical.
In an n-ary tree, more than two operands may share the same parent operator node, and their evaluation order can be rearranged to eliminate more redundancies.
Therefore, multiple binary subexpressions may be identified as redundant within an n-ary expression, and some of them can conflict with each other.
Specifically, the expression $e=x\oplus y \oplus\cdots$ in an n-ary tree is denoted by \(\oplus(x,y,\cdots)\).
If both $x\oplus y$ and $y\oplus z$ in the expression $\oplus(x,y,z)$
are redundant, only one of them can belong to the solution.

To model the conflict relationship between binary subexpression candidates.
We define the \emph{Pair Graph}, an undirected conflict graph \(G=(V,E)\) to formalize the constraint. 
$V$ contains all possible binary subexpressions.
In particular, it generates $\binom{n}{2}$ binary subexpressions for each operator node with $n$ leaf children.
A node and an expression are equivalent in $G$ and used interchangeably.
Each edge indicates the conflict relation between two nodes, which means
that the two expressions (nodes) share the same array reference and cannot be extracted simultaneously.
Thus, any legal solution is an independent set of $G$.
Figure~\ref{fig:example} shows the corresponding Pair Graph
to the flattened tree.
The formal definition of $G$ is as follows:

\[V=\bigcup\limits_{e=\oplus(\cdots)}\bigcup\limits_{x,y \in e}\{ x\oplus y \mid x,y \text{ are leaves}, x \neq y\},\]

\[E=\{(v,v') \mid v,v' \in V, v=x\oplus y \neq v'=x'\oplus y', \{x,y\} \cap \{x',y'\} \neq \emptyset \}.\]

In each iteration of the detection algorithm, finding redundancies corresponds to selecting a set of non-conflicting binary subexpression candidates, which forms an independent set of the Pair Graph.
To determine the optimal solution, an objective function on the graph needs to be defined.
The loop ranges are usually not known at the compiling time.
So it cannot precisely evaluate the redundant computations
between conflict redundancies.
For example, given three expressions
$\oplus(x,y,z)$, $\oplus(x',y')$ and $\oplus(y'',z'')$, if $\eri(x\oplus y)=\eri(x'\oplus y')$ and $\eri(y\oplus z)=\eri(y''\oplus z'')$,
it is hard to determine which one incurs more redundant computations before the execution.

Therefore, we use an approximate metric.
If $n$ expressions with the same $\eri$ are extracted, the new expression trees have
$n$ fewer binary expressions that are replaced with auxiliary array loads.
The reduced number of redundant computations on the expression trees can be
approximated to $n-1$, where the minus one denotes the precomputation of the auxiliary array.
Let $S\subseteq V$ be a solution and $\eri(S)$ be the set of all $\eri$ values of expressions in $S$, i.e.,
\[\eri(S)=\{\eri(v) \mid v \in S\}.\]

The optimal solution in a single detection iteration is tied to the objective function:
\begin{equation}
    \begin{aligned}
    \argmax_{S \in I_G} |S| - |\eri(S)|,
    \end{aligned}
    \label{eq:ori-obj}
\end{equation}
where $I_G$ denotes the set of all independent sets of graph $G$.

This problem can be reduced to the Maximum Independent Set (MIS) problem by augmenting the Pair Graph $G$ to $\bar{G}=(\bar{V}, \bar{E})$.
For each $\eri$ value $k$, $\bar{G}$ adds an auxiliary node $\bar{v}_k$ connected to all nodes with $\eri$ value $k$ in $G$, i.e., $\bar{V}=V \cup V'$ where $V'=\{\bar{v}_k \mid k \in \eri(V)\}, \eri(\bar{v}_k)=k$, 
and $\bar{E}=E \cup \{(\bar{v}_k, v) \mid \bar{v}_k \in V', v \in V, \eri(v)=k\}$.
The formal description is in Theorem~\ref{theorem:MIS}.

\begin{theorem}
    The optimization problem on graph $G$ defined in Equation~\ref{eq:ori-obj}
    can be reduced to the MIS problem on $\bar{G}$.
    \label{theorem:MIS}
\end{theorem}

\begin{minipage}{0.67\textwidth}
  \textit{Proof \wzx{sketch}}.
  We first define a bijective function $f$ that maps \wzx{each} independent set $S$ of $G$ to an independent set $\bar{S}$ of $\bar{G}$ by adding auxiliary nodes. An illustration is shown on the right.
  Then we show that every MIS of $\bar{G}$ lies in the range of $f$, which establishes a correspondence between the solution space of the original problem on $G$ and that of the MIS problem on $\bar{G}$.
  Finally, two objective functions are introduced for the optimization problems on $G$ and $\bar{G}$, respectively, and we prove their equivalence under $f$. The full proof is provided in Appendix~\ref{appendix:proof}.
\end{minipage}
\hfill
\begin{minipage}{0.32\textwidth}
  \centering
  \begin{tikzpicture}[yscale=0.4, scale=0.6]
      \coordinate (UL) at (0,0);
      \coordinate (UR) at (0,-8);
      \coordinate (DL) at (0,0);
      \coordinate (DR) at (0,-8);

      \node[] at ($(UL)+(-2,0)$){\small $G$};
      \node[] at ($(UR)+(-2,-2)$){\small $\bar{G}$};
      
      \foreach \i in {1,...,5}{
          \node[empty vertex] (ul\i) at ($(UL)+(0-72*\i:1.4)$) {};
          \node[empty vertex] (dl\i) at ($(DL)+(0-72*\i:1.4)$) {};
          \node[empty vertex] (dr\i) at ($(DR)+(0-72*\i:1.4)$) {};
          \node[empty vertex] (ur\i) at ($(UR)+(0-72*\i:1.4)$) {};
      }

      \foreach \i in {6,...,8}{
          \node[empty vertex] (ur\i) at ($(UR)+(0.3,-2.5)+(60-72*\i:1)$) {};
          \node[empty vertex] (dr\i) at ($(DR)+(0.3,-2.5)+(60-72*\i:1)$) {};
      }

      \coordinate (C1) at ($(UR)+(0.0,0.1)$);

      \foreach \i/\ang/\rad in {1/0/2.5, 2/90/2.2, 3/180/2.5, 4/270/2.2} {
          \coordinate (P\i) at ($(C1)+(\ang+25:\rad)$);
      }
      \filldraw[
          fill=gray!40,
          fill opacity=0.25,
          draw=gray,
          line width=1pt,
      ] (P1) -- ($(P2)+(0,0)$) -- (P3) -- ($(P4)+(0,0)$) -- cycle;

      \coordinate (C2) at ($(UR)+(0.2,-3.1)$);

      \foreach \i/\ang/\rad in {1/0/1.7, 2/90/1.4, 3/180/1.7, 4/270/1.4} {
          \coordinate (P\i) at ($(C2)+(\ang+24:\rad)$);
      }

      \filldraw[
          fill=gray!40,
          fill opacity=0.25,
          draw=gray,
          line width=1pt,
      ] (P1) -- ($(P2)+(0,0)$) -- (P3) -- ($(P4)+(0,0)$) -- cycle;

      \node[filled vertex, fill=green] (dl1) at (dl1){};
      \node[filled vertex, fill=red] (dl4) at (dl4){};
      \node[opac filled vertex, fill=red] (dl5) at (dl5){};
      \node[opac filled vertex, fill=yellow] (dl2) at (dl2){};
      \node[opac filled vertex, fill=yellow] (dl3) at (dl3){};

      \draw[] ($(dl4)!0.5!(dl1)$) ellipse (10pt and 60pt) node[xshift=10pt,yshift=10pt]{\small $S$};

      \node[filled vertex, fill=green] (dr1) at (dr1){};
      \node[filled vertex, fill=red] (dr4) at (dr4){};
      \node[opac filled vertex, fill=red] (dr5) at (dr5){};
      \node[opac filled vertex, fill=yellow] (dr2) at (dr2){};
      \node[opac filled vertex, fill=yellow] (dr3) at (dr3){};

      \node[opac filled vertex, fill=red] (dr6) at (dr6){};
      \node[ax vertex, pattern color=black!80] (dr6) at (dr6){};
      \node[opac filled vertex, fill=green] (dr7) at (dr7){};
      \node[ax vertex, pattern color=black!80] (dr7) at (dr7){};
      \node[filled vertex, fill=yellow] (dr8) at (dr8){};
      \node[ax vertex] (dr8) at (dr8){};

      \draw[]
          plot[smooth, tension=1]
          coordinates {
              ($(dr4)+(-0.2,0.2)$)
              ($(dr4)+(0.2,0.2)$)
              ($(dr1)+(0.2,-0.2)$)
              ($(dr8)+(-0.2,-0.2)$)
              ($(dr4)+(-0.2,0.2)$)
          };
      \node[] at ($(dr4)+(0.5,0.5)$){\small $\bar{S}$};

      \draw[<->, thick] ($(DL)!0.35!(DR)$) -- ($(DL)!0.7!(DR)$) node[midway,left]{\small $f(S)=\bar{S}$};

      \node[] at ($(ur5)+(0.7,-0.3)$){\small $V$};
      \node[] at ($(ur6)+(0.6,-0.6)$){\small $V'$};

      \foreach \i/\j in {3/1,3/5,1/5,2/4}{
          \draw[dashed, draw opacity=0.5] (dl\i) -- (dl\j);
          \draw[dashed, draw opacity=0.5] (dr\i) -- (dr\j);
      }

      \foreach \i/\j in {6/4,6/5,7/1,8/2,8/3}{
          \draw[dashed, draw opacity=0.5] (dr\i) -- (dr\j);
      }
  \end{tikzpicture}
\end{minipage}

\subsection{Inner-Dimension-First Heuristic}
\label{subsec:dimension-first}

It is well-known that the MIS problem is NP-hard.
For a graph with $n$ nodes, the worst-case time complexity is $O(2^n)$.
For expression trees with numerous operands, exhaustive search quickly becomes computationally infeasible.
In addition, Equation~\ref{eq:ori-obj} on the Pair Graph may admit multiple optimal solutions. Different solutions correspond to different sets of auxiliary arrays, which in turn can lead to different performance of the generated code.
Auxiliary arrays whose reuses occur along the innermost loop usually exhibit better locality and facilitate vectorization.
We therefore introduce a heuristic strategy that reduces the problem size and improves the quality of the chosen auxiliary arrays by prioritizing expressions that satisfy specific conditions.

The array's memory access patterns in real applications are typically regular and symmetric. As described in Section~\ref{subsec:optimization}, the auxiliary arrays can be compressed to lower dimensions to save memory and improve the cache performance. Based on the observation, we propose a heuristic algorithm called inner-dimension-first (IDF) extraction, which aims to find solutions with a higher potential for compression while reducing the problem size. This approach increases the opportunity to leverage cache locality and vectorization.
The intuition behind this approach can be illustrated by the example in Figure~\ref{fig:mgrid}.
For instance, if we choose $R(i-1,k,j)+R(i+1,k,j)$ as an auxiliary array in the first iteration of our algorithm, it cannot be compressed to a one-dimensional array because its reuses are along the outer loop dimension $j$ and $k$. Furthermore, other solutions like $R(i-1,k-1,j)+R(i-1,k,j-1)$ are irregular and may result in discrete memory access, leading to poor locality. In contrast, if $R(i,k,j-1)+R(i,k,j+1)$ is chosen, it can be reused along $i$ dimension.

IDF extraction is a greedy algorithm that prioritizes selecting expressions with zero delta in the innermost dimension, i.e., $e.\mathit{exprDelta}[i_n] = 0$. A zero delta in the innermost dimension indicates that the reuses are more likely to occur along the innermost $i_n$ dimension, e.g., any binary subexpression of the $\mathrm{aa\_1\_0}$'s $\mathit{exprDelta}[i]=0$ and $\mathrm{aa\_1\_0}$ is reused along $i$ loop in Figure~\ref{fig:mgrid}.
The algorithm runs in a try-until manner. First, it only generates the nodes with $\mathit{exprDelta}[i_n]=0$ when constructing Pair Graph. This forms a subgraph of the original Pair Graph, containing significantly fewer nodes. If any solutions are found, the algorithm accepts them and proceeds to the next extraction iteration. If no solution is found, the algorithm will generate a new subgraph containing only nodes with $\mathit{exprDelta}[i_{n-1}]=0$ and try to find the solution of this graph, repeating this process until it finds a solution at some dimension $i$.

Although IDF does not change the exponential worst-case time complexity $O(2^n)$, it speeds up solving by significantly reducing the number of nodes $n$ in the Pair Graph. In the case of Figure~\ref{fig:mgrid}, the original Pair Graph would contain more than one hundred nodes, making exhaustive search infeasible, whereas the IDF strategy reduces the search to a subgraph with only 18 nodes and solves it in negligible time.
In addition, other heuristic or approximate methods for MIS could also be incorporated into the solver, but such extensions are beyond the scope of this paper.

\section{Implementation and Limitations}
\label{sec:threats}

We implemented RACE based on the Flang and Clang frontends in LLVM infrastructure. Our tool provides multiple options to control different levels of optimization aggressiveness. To invoke the transformation, developers insert compilation directives before and after the target loop nest, after which RACE performs a source-to-source transformation based on the abstract syntax tree (AST) produced by the frontend.
RACE parses the annotated loop nests and extracts the candidate expression trees. Then it applies the proposed algorithm to detect redundant computations and rewrites the expression trees accordingly. After the detection finished, it compresses the auxiliary arrays and generates the transformed code based on the rewritten expression trees.

RACE currently targets a restricted but practically important class of loop nests.
In this paper, we focus on perfectly nested loops without internal control flow, redundant computations over unmodified arrays, and relatively simple array references of the form
$A[a_1 i_{s_1}+b_1]\cdots[a_n i_{s_n}+b_n]$.
Some simple non-perfectly nested loops, such as those with initialization statements or with branches containing side-effect-free statements, may still be compatible with the RACE mechanism. However, extending the approach to general non-perfectly nested loops requires a complete legality analysis and formalization, which is left for future work.

RACE is an annotation-guided aggressive optimization.
It uses developer annotations both to mark target loops and to assert semantic assumptions.
For example, the method assumes that targeted function calls such as $\sin$ and $\cos$ are pure and deterministic, i.e., they have no side effects and always return the same floating-point value for the same input.
Although compiler analyses may help check some of these assumptions, full automation remains difficult in the presence of issues such as pointer aliasing and unknown call targets.
When reassociation is enabled, the optimization is provided on an "as-is" basis: arithmetic operations may be reordered, and numerical stability is not guaranteed.
It is therefore the developer's responsibility to enable these aggressive modes only when such changes are acceptable.

Another limitation is the storage overhead of the auxiliary arrays introduced by the transformation.
Although the contraction scheme in Section~\ref{subsec:optimization} can substantially reduce the footprint of temporary arrays, in extreme cases, the additional memory demand may exceed the available memory budget and lead to an out-of-memory failure.
This extra storage may also increase memory traffic and cache pressure, so reducing arithmetic operations does not always lead to a speedup.
Therefore, RACE is exposed as an opt-in optimization through annotations, so that programmers can decide whether the optimization is appropriate for a given loop.

\section{Evaluation}
\label{sec:experiment}
\subsection{Experiment Setup}
We evaluated performance test in SPEC applications and extracted the kernels for more detailed performance analysis.
The performance tests in SPEC applications were conducted on two different platforms: an Intel Xeon Gold 6467C 30-core CPU @ 3.40 GHz and an AMD EPYC 9654 96-core Processor @ 2.40 GHz. The detailed performance analysis of the extracted kernels was only conducted on the Intel platform due to the limited space in the paper and the similarity between the platforms.
The SPEC applications were run using the standard configurations and input data provided by SPEC suite. The benchmarks were compiled by gfortran-9.4.0 or gcc-9.4.0 with compilation flags "-O3" to enable all standard optimizations.

\subsection{Benchmark and Redundancy Analysis}

\begin{table}[h]
  \centering
  \setlength{\tabcolsep}{4pt}
  \caption{Benchmark optimization table. \textit{Reduced Ops} is the fraction of the run-time reduced arithmetic operations(including $\sin/\cos$). \textit{AA Num} is the number of auxiliary arrays found in total. \textit{Alg Iter} is the iteration number of our algorithm to find all the auxiliary arrays. \textit{Operations} are static operation analysis. RACE-NR refers to the use of our algorithm without reassociation.}
  \begin{tabular}{|c|c|c|c|c|c|c|c|c|c|}
    \hline
  \multirow{2}{*}{{App}} & \multirow{2}{*}{{Case}} & \multirow{2}{*}{{\small \makecell{Reduced \\ Ops}}} & \multirow{2}{*}{{\small \makecell{AA \\ Num}}}  
                              & \multirow{2}{*}{{\small \makecell{Alg \\ Iter}}}  & \multicolumn{5}{c|}{{Operations(Base/RACE-NR/RACE)}}  \\ \cline{6-10} 
                               & & & &   & add & sub & mul & div & sin/cos \\ \hline
  \multirow{4}{*}{POP}     & hdifft\_gm   & 0.63 & 2        & 1 & 14/11/4     & - & - & - & - \\
            & calc\_tpoints& 0.55 & 9        & 3 & 9/9/6      & - & 11/5/5 & - & 16/4/4     \\
            & ocn\_export  & 0.17 & 2        & 1 & 1/1/1      & 1/1/1      & 6/6/5      & 2/2/1      & 4/2/2      \\ \hline
  \multirow{5}{*}{WRF}     & rhs\_ph1     & 0.06 & 3        & 2 & 6/5/5      & 9/9/9      & 12/10/10    & 2/2/2      & - \\ 
              & rhs\_ph2     & 0.16 & 3        & 2 & 6/5/5      & 9/9/9      & 12/10/10    & 2/2/2      & - \\ 
              & diffusion1 & 0.44 & 20       & 5 & 18/18/8     & 6/4/4      & 26/21/15    & 4/3/2      & - \\ 
              & diffusion2 & 0.60 & 19       & 5 & 18/16/8     & 6/4/4      & 26/20/14    & 4/3/2      & - \\ 
              & diffusion3 & 0.49 & 19       & 6 & 10/6/6     & 6/4/4      & 32/18/17    & 2/1/1      & - \\ \hline
  \multirow{3}{*}{mgird}   & psinv        & 0.38 & 9        & 3 & 27/23/13    & - & 4/4/6 & - & - \\
              & resid        & 0.45 & 4        & 3 & 23/19/11     & 4/4/4      & 4/4/4      & - & - \\
              & rprj3        & 0.19 & 5        & 2 & 26/26/20    & - & 4/4/4 & - & - \\ \hline
  \multirow{4}{*}{stencil}   & gaussian        & 0.43 & 13        & 4 & 24/24/16    & - & 25/6/11 & 1/1/1 & - \\
              & j3d27pt        & 0.35 & 20        & 3 & 26/26/18     & -      & 27/15/15      & 1/1/1 & - \\
              & poisson        & 0.37 & 3        & 2 & 16/15/8     & 2/2/2      & 3/3/3      & - & - \\
              & derivative     & 0.71 & 86        & 11 & 99/54/45     & 96/24/16      & 297/101/76      & - & - \\ \hline
  \end{tabular}
  \label{tab:benchmark}
\end{table}

We selected 11 cases in three different applications and 4 stencil kernels to evaluate the performance of the code generated by our tool.
For each case above, RACE completes the source-to-source transformation in less than 1.5 seconds, showing the efficiency of our algorithm.
Table~\ref{tab:benchmark} shows the information of the selected cases. All of the loops in the cases are floating-point computations. The column label \textit{Reduced Ops} represents reduction in run-time arithmetic operations. \textit{Operations} column contains static analysis of operations, where the numbers are the counts of arithmetic operations in one iteration of the innermost loop.

POP is a high-performance ocean modeling software used for climate simulation and ocean circulation modeling. WRF is a numerical weather prediction model used for atmospheric research and operational forecasting. mgrid is a benchmark used to simulate multigrid algorithms for solving partial differential equations.
The remaining cases are stencil kernels widely used in scientific applications. \textit{Poisson}, \textit{derivative} and \textit{j3d27pt} are 3D stencil kernels whose grid size is set to $100^3$ in our experiment, while \textit{gaussian} are 2D stencil kernels with grid size set to $500^2$. 
The POP and WRF kernels were evaluated using their official implementations in the SPEC CPU 2017 benchmark suite, and the mgrid kernel was executed through the SPEC CPU 2000 suite.

Our algorithm effectively detects various redundant arithmetic operations. The $\sin/\cos$ operations are reduced from $16$ to only $4$ in the case \textit{calc\_tpoints} of POP. A large number of auxiliary arrays are found in WRF's cases with few iterations, leading to nearly a half reduction in arithmetic operations. RACE can also handle large expression trees, such as the \textit{derivative} case, where 86 auxiliary arrays are found and the multiplications are reduced from $297$ to $76$.

The detection algorithm on n-ary trees can more effectively identify redundancies. If only the algorithm on binary trees is used (which is denoted as RACE-NR in Table~\ref{tab:benchmark}), the case \textit{diffusion1} and \textit{diffusion2} can only identify 8 and 9 auxiliary arrays, respectively, with a reduction of 8 and 11 redundant operations. In contrast, using the n-ary tree-based algorithm enables the identification of 20 and 19 auxiliary arrays, with a reduction of 25 and 26 redundant operations, respectively.

Cases in mgrid and stencil kernels involve many array references with the same $\rpi$, leading to multiple solutions when solving the optimal reuse problem defined in Section~\ref{subsec:nary-algorithm}. \wzx{The inner-dimension-first method proposed} in Section~\ref{subsec:dimension-first} identifies the auxiliary arrays with innermost loop reuse, which are more cache-friendly. 
In case \textit{rprj3}, the loop variables inside the array subscripts are multiplied by a factor of 2. Although the names of the array references are the same, they have different $\rpi$ values. The redundancies are successfully detected, demonstrating the effectiveness of our two-level hash algorithm.

\subsection{\wzx{Speedup Results}}

\begin{figure}[!htb]
  \centering
  \includegraphics[width=\linewidth]{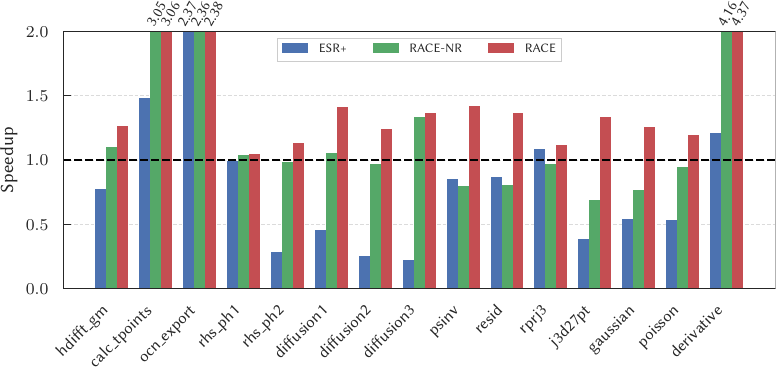}
  \caption{The speedup of different optimization methods relative to the original code on the Intel CPU. Among the legend, RACE-NR refers to the use of our algorithm without reassociation; RACE uses the full abilities which involve reassociation, inner-dimension-first extraction, and array contraction. Note that the ESR+ represents ESR+reassociation.}
  \label{fig:perf-intel}
\end{figure}

\begin{figure}[!htb]
  \centering
  \includegraphics[width=\linewidth]{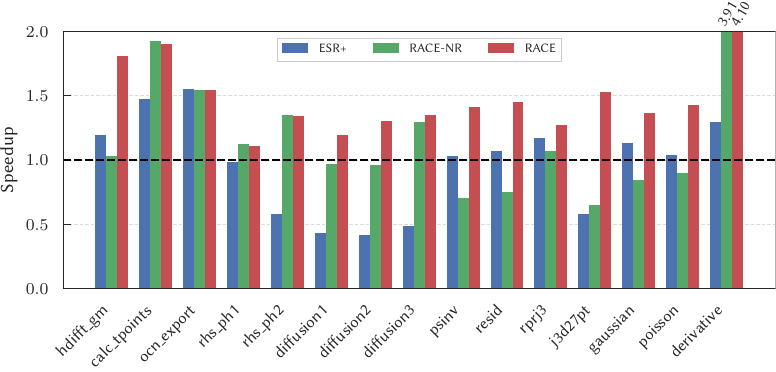}
  \caption{The speedup of different optimization methods relative to the original code on the AMD CPU. The legend is the same as in Figure~\ref{fig:perf-intel}.}
  \label{fig:perf-amd}
\end{figure}

We evaluated the single-core performance of the code generated by our tool. Results for SPEC benchmarks were measured within the applications using the standard SPEC run configuration. Since our algorithm targets loop-level optimization, we therefore report the loop-level speedup. Figure~\ref{fig:perf-intel} and Figure~\ref{fig:perf-amd} show the speedup compared to the baseline (original) code and ESR+ (ESR+reassociation) on the Intel and AMD CPUs, respectively.

Our results demonstrate that RACE achieves significant performance improvements across most benchmarks. RACE outperforms ESR+ in 14 out of 15 cases on both platforms, except for the \textit{ocn\_export} case where RACE generates the same code as ESR+.
The RACE-NR method, despite not using reassociation, still provides notable improvements in cases like \textit{calc\_tpoints}, \textit{derivative}.
The ESR+ method shows several performance regressions. This is because the auxiliary variables extracted by ESR+ are scalars placed in the innermost loop and reused in a pipelined manner, which hinders the compiler's vectorization optimization. RACE can retain the auxiliary arrays as multi-dimensional arrays, which facilitates vectorization.

The overall performance gain is influenced by both the amount and types of detected redundancies, as detailed in Table~\ref{tab:benchmark}. Fewer auxiliary arrays are identified in the \textit{rhs\_ph} case than in \textit{diffusion}, explaining the lower speedup for \textit{rhs\_ph}. In contrast, the \textit{derivative} case contains many redundant operations—up to 86 auxiliary arrays and 221 eliminated multiplications—leading to a significant speedup on both platforms. The intensity of redundant arithmetic operations largely determines the performance gain. On the Intel CPU, the speedup reaches 3.06$\times$ in the \textit{calc\_tpoints} case of POP, where the $\sin/\cos$ functions are computation-intensive. A similar pattern is observed in \textit{ocn\_export}, where the reduction in operations is small but the speedup remains notable.

However, for loops that are memory-access intensive, such as the \textit{rprj3} case, the optimization is less effective. Moreover, reducing arithmetic intensity may increase memory pressure. While auxiliary arrays reduce arithmetic operations, they introduce additional memory allocations and increase the memory access footprint. The newly allocated arrays increase the cache working set, potentially leading to more cache misses. Compressing the dimensions of auxiliary arrays using the method described in Section~\ref{subsec:optimization} helps mitigate this issue.

\subsection{\wzx{Scaling Results}}

To gain deeper insights into the performance characteristics of the cases, we extracted the representative kernels from the SPEC applications and converted them into standalone benchmarks for detailed performance analysis. The performance analysis focuses on the runtime and memory data volume (the bytes transferred between the last-level cache and DRAM) for different problem sizes and the parallel scalability.

\subsubsection{Runtime Comparison}

\begin{figure}[!htb]
  \centering
  \includegraphics[width=\textwidth]{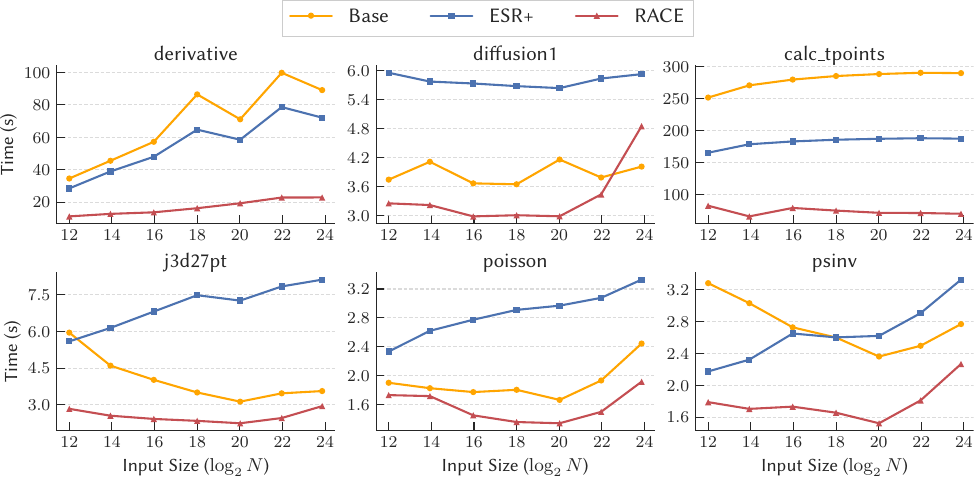}
  \caption{Runtime vs. Input Size on Intel CPU. The total computation amounts ($N \times T = 2^{31}$) are the same for different kernels and input sizes.}
  \label{fig:runtime_vs_size}
\end{figure}

Figure~\ref{fig:runtime_vs_size} presents the runtime of six representative kernels on the Intel CPU for varying input sizes ($\log_{2} N$). For all experiments, the product of input size and time step ($N \times T = 2^{31}$) is fixed to ensure a fair comparison across different kernels and input sizes.

Overall, RACE achieves the lowest runtime, except for \textit{diffusion1} at size $N=2^{24}$, consistent with the previous results.
The runtime trends vary across different kernels. For compute-intensive kernels such as \textit{derivative} and \textit{calc\_tpoints}, the runtime of RACE increases more slowly as problem size increases. In contrast, for less compute-intensive kernels such as \textit{diffusion1} and \textit{psinv}, the runtime grows more steeply at larger input sizes, suggesting that the memory access cost becomes dominant, which conceals part of the performance gain from reducing arithmetic operations.

\wzx{Although the total computation amounts are kept the same, the runtime of certain kernels can be shorter at medium input sizes than at small input sizes, due to the overhead of handling loop remainder from vectorization, which is more pronounced when the input grid is small. }

\subsubsection{Memory Volume Saved by Array Contraction}

\begin{figure}[!htbp]
  \centering
  \includegraphics[width=\textwidth]{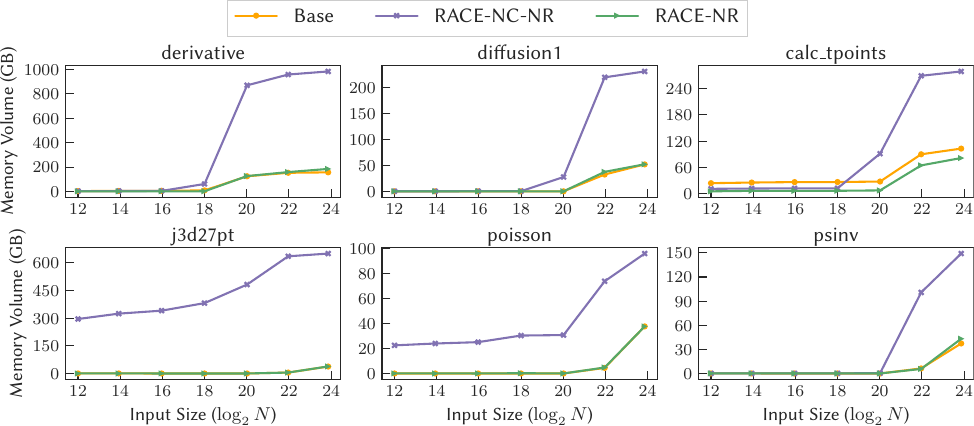}
  \caption{Memory Volume vs. Input Size on Intel CPU. RACE-NC-NR denotes the RACE method without array contraction and reassociation. RACE-NR uses array contraction while remaining reassociation disabled.}
  \label{fig:mem_volume}
\end{figure}

\begin{figure}[!hbtp]
  \centering
  \includegraphics[width=\textwidth]{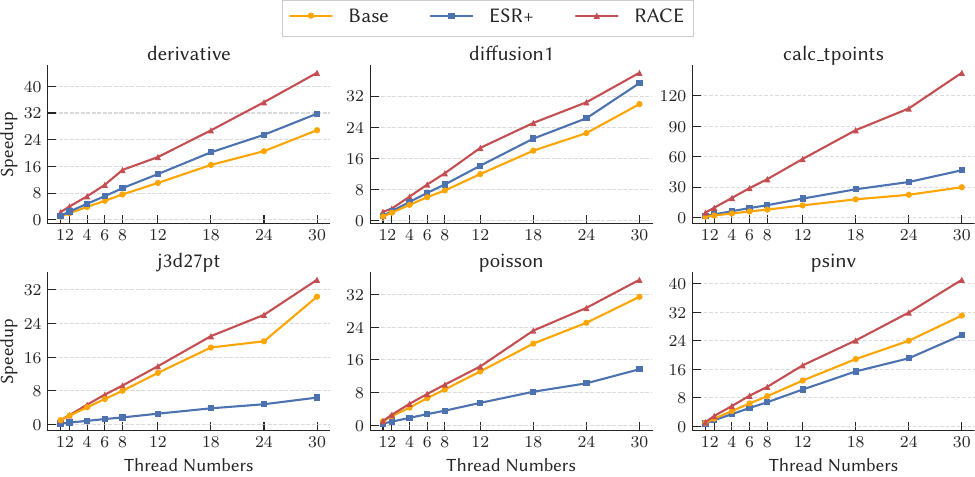}
  \caption{Multithreaded Scalability on Intel CPU. The speedup is computed as $t_{1,\text{Base}}/t_{n,\text{opt}}$, where $t_{n,\text{opt}}$ denotes the runtime of opt-optimized code using $n$ threads. $t_{1,\text{Base}}$ is the runtime of the original code using one thread. All kernel and thread number configurations were run with $N=180^3,T=500$.}
  \label{fig:scalability}
\end{figure}

The run-time memory volume is measured to demonstrate the effectiveness of array contraction proposed in Section~\ref{subsec:optimization}. Figure~\ref{fig:mem_volume} shows the memory volume transferred between the last-level cache and DRAM for different input sizes on the Intel CPU. RACE-NC-NR denotes the RACE method without array contraction and reassociation. The only difference between RACE-NC-NR and RACE-NR is that the former disables array contraction while the latter enables it. This allows us to isolate the impact of array contraction from other factors affecting memory volume.

The results show that the memory volume of RACE-NC-NR is significantly higher than that of RACE-NR, while RACE-NR is close to the baseline code. This indicates that array contraction effectively alleviates the memory pressure introduced by auxiliary arrays, especially at larger input sizes, leading to substantial performance improvements. The memory volume is much higher at large input sizes because the input grids exceed the cache capacity.

\subsubsection{Multi-threaded Scalability}

We also evaluate the OpenMP multithreaded scalability of our approach, which is shown in Figure~\ref{fig:scalability}. The speedup is compared against the baseline code using a single thread. The results indicate that RACE achieves nice scalability and outperforms the other methods consistently across different thread counts. The introduced auxiliary arrays do not hinder multi-threaded performance.

\section{Conclusion}
\label{sec:conclusion}

This paper proposes RACE, a two-level hashing framework for detecting redundant array computations in nested loops. The first-level hash uses $\rpi$ to determine whether two array references share the same reference pattern.
The second-level hash computes $\eri$ for each binary expression as an attribute of the expression, encoding expression-level reuse without explicitly comparing different expressions.
Using $\eri$, we give a linear-time algorithm on binary expression trees and discuss code generation and optimization for the resulting auxiliary arrays.
Furthermore, tree flattening and the Pair Graph are introduced to identify redundancies in n-ary trees. The problem is formulated and reduced to a maximum independent set problem. To efficiently find a solution, we propose a memory-friendly inner-dimension-first heuristic that prioritizes selecting auxiliary arrays with inner-loop reuse.
Experiments show that RACE achieves significant speedups across various benchmarks, demonstrating the effectiveness of the algorithm.

\begin{acks}
This work is supported by National Key R\&D Program of China under Grant No. 2025YFB3003604, 2023YFB3001700, National Natural Science Foundation of China under Grant No. 62372432.
\end{acks}

\section*{Data-Availability Statement}

The artifact for the experiments is available on Zenodo~\cite{RACE-artifact}.
It is provided as a Docker image archive that includes all programs and scripts required to reproduce the experimental results.

\bibliographystyle{ACM-Reference-Format}
\bibliography{references}

\appendix

\section{Proof of Theorem~\ref{theorem:MIS}}
\label{appendix:proof}

\begin{figure}[htbp]
    \centering
    \begin{tikzpicture}[yscale=0.4]
        \coordinate (UL) at (0,0);
        \coordinate (UR) at (5,0);
        \coordinate (DL) at (0,-8);
        \coordinate (DR) at (5,-8);

        \node[] at ($(UL)+(0,3)$){$G$};
        \node[] at ($(UR)+(0,3)$){$\bar{G}$};
        \foreach \i in {1,...,5}{
            \node[empty vertex] (ul\i) at ($(UL)+(0-72*\i:1.4)$) {};
            \node[empty vertex] (dl\i) at ($(DL)+(0-72*\i:1.4)$) {};
            \node[empty vertex] (dr\i) at ($(DR)+(0-72*\i:1.4)$) {};
            \node[empty vertex] (ur\i) at ($(UR)+(0-72*\i:1.4)$) {};
        }

        \foreach \i in {6,...,8}{
            \node[empty vertex] (ur\i) at ($(UR)+(0.3,-2.5)+(60-72*\i:1)$) {};
            \node[empty vertex] (dr\i) at ($(DR)+(0.3,-2.5)+(60-72*\i:1)$) {};
        }

        \coordinate (C1) at ($(UR)+(0.0,0.1)$);

        \foreach \i/\ang/\rad in {1/0/2.5, 2/90/2.2, 3/180/2.5, 4/270/2.2} {
            \coordinate (P\i) at ($(C1)+(\ang+25:\rad)$);
        }
        \filldraw[
            fill=gray!40,
            fill opacity=0.25,
            draw=gray,
            line width=1pt,
        ] (P1) -- ($(P2)+(0,0)$) -- (P3) -- ($(P4)+(0,0)$) -- cycle;

        \coordinate (C2) at ($(UR)+(0.2,-3.1)$);

        \foreach \i/\ang/\rad in {1/0/1.7, 2/90/1.4, 3/180/1.7, 4/270/1.4} {
            \coordinate (P\i) at ($(C2)+(\ang+24:\rad)$);
        }

        \filldraw[
            fill=gray!40,
            fill opacity=0.25,
            draw=gray,
            line width=1pt,
        ] (P1) -- ($(P2)+(0,0)$) -- (P3) -- ($(P4)+(0,0)$) -- cycle;

        \node[filled vertex, fill=green] (ul1) at (ul1){};
        \node[filled vertex, fill=red] (ul4) at (ul4){};
        \node[filled vertex, fill=red] (ul5) at (ul5){};
        \node[filled vertex, fill=yellow] (ul2) at (ul2){};
        \node[filled vertex, fill=yellow] (ul3) at (ul3){};

        \node[filled vertex, fill=green] (ur1) at (ur1){};
        \node[filled vertex, fill=red] (ur4) at (ur4){};
        \node[filled vertex, fill=red] (ur5) at (ur5){};
        \node[filled vertex, fill=yellow] (ur2) at (ur2){};
        \node[filled vertex, fill=yellow] (ur3) at (ur3){};

        \node[filled vertex, fill=red] (ur6) at (ur6){};
        \node[ax vertex] (ur6) at (ur6){};
        \node[filled vertex, fill=green] (ur7) at (ur7){};
        \node[ax vertex] (ur7) at (ur7){};
        \node[filled vertex, fill=yellow] (ur8) at (ur8){};
        \node[ax vertex] (ur8) at (ur8){};

        \node[filled vertex, fill=green] (dl1) at (dl1){};
        \node[filled vertex, fill=red] (dl4) at (dl4){};
        \node[opac filled vertex, fill=red] (dl5) at (dl5){};
        \node[opac filled vertex, fill=yellow] (dl2) at (dl2){};
        \node[opac filled vertex, fill=yellow] (dl3) at (dl3){};

        \draw[] ($(dl4)!0.5!(dl1)$) ellipse (10pt and 60pt) node[xshift=15pt,yshift=15pt]{$S$};

        \node[filled vertex, fill=green] (dr1) at (dr1){};
        \node[filled vertex, fill=red] (dr4) at (dr4){};
        \node[opac filled vertex, fill=red] (dr5) at (dr5){};
        \node[opac filled vertex, fill=yellow] (dr2) at (dr2){};
        \node[opac filled vertex, fill=yellow] (dr3) at (dr3){};

        \node[opac filled vertex, fill=red] (dr6) at (dr6){};
        \node[ax vertex, pattern color=black!80] (dr6) at (dr6){};
        \node[opac filled vertex, fill=green] (dr7) at (dr7){};
        \node[ax vertex, pattern color=black!80] (dr7) at (dr7){};
        \node[filled vertex, fill=yellow] (dr8) at (dr8){};
        \node[ax vertex] (dr8) at (dr8){};

        \draw[]
            plot[smooth, tension=1]
            coordinates {
                ($(dr4)+(-0.2,0.2)$)
                ($(dr4)+(0.2,0.2)$)
                ($(dr1)+(0.2,-0.2)$)
                ($(dr8)+(-0.2,-0.2)$)
                ($(dr4)+(-0.2,0.2)$)
            };
        \node[] at ($(dr4)+(0.4,0.3)$){$\bar{S}$};


        \draw[<->, thick] ($(DL)!0.35!(DR)$) -- ($(DL)!0.7!(DR)$) node[midway,above]{$f(S)=\bar{S}$};

        \node[] at ($(ur5)+(0.7,0)$){$V$};
        \node[] at ($(ur6)+(0.5,-0.5)$){$V'$};

        \foreach \i/\j in {3/1,3/5,1/5,2/4}{
            \draw (ul\i) -- (ul\j);
            \draw (ur\i) -- (ur\j);
            \draw[dashed, draw opacity=0.5] (dl\i) -- (dl\j);
            \draw[dashed, draw opacity=0.5] (dr\i) -- (dr\j);
        }

        \foreach \i/\j in {6/4,6/5,7/1,8/2,8/3}{
            \draw (ur\i) -- (ur\j);
            \draw[dashed, draw opacity=0.5] (dr\i) -- (dr\j);
        }
    \end{tikzpicture}
    \caption{The above two graphs illustrate an example of the graph $G$ and the corresponding auxiliary graph $\bar{G}$.
    The bottom two graphs highlight an independent set $S$ of $G$ and its corresponding independent set $\bar{S}=f(S)$ of $\bar{G}$.}
    \label{fig:proof-example}
\end{figure}

The objective function defined in Equation~\ref{eq:ori-obj} is

\begin{equation*}
    \begin{aligned}
    \argmax_{S \in I_G} |S| - |\eri(S)|,
    \end{aligned}
\end{equation*}
where $I_G$ denotes the set of independent sets of graph $G$.

To reduce this problem to a Maximum Independent Set (MIS) problem, an auxiliary graph $\bar{G}=(\bar{V}, \bar{E})$ is constructed.
For each $\eri$ value $k$, $\bar{G}$ adds an auxiliary node $\bar{v}_k$ connected to all nodes with $\eri$ value $k$ in $G$, i.e., $\bar{V}=V \cup V'$ where $V'=\{\bar{v}_k \mid k \in \eri(V)\}, \eri(\bar{v}_k)=k$, and $\bar{E}=E \cup \{(\bar{v}_k, v) \mid \bar{v}_k \in V', v \in V, \eri(v)=k\}$.
An example of graph $G$ and the corresponding auxiliary graph $\bar{G}$ are shown in Figure~\ref{fig:proof-example}.
The formal reduction description is in Theorem~\ref{theorem:MIS}. To prove the theorem, we first introduce the following function and lemmas.

Define transformation function $f: I_G \to f(I_G)$ as \[f(S) = S \cup \{\bar{v}_k \in V' \mid k \notin \eri(S)\}.\] We have lemma~\ref{lemma:bijective}.

\begin{lemma}
    $f$ is a bijective function.
    \label{lemma:bijective}
\end{lemma}

\begin{proof}
  It can be seen that $f$ adds auxiliary nodes of $\bar{G}$ with disappeared $\eri$ values into $S$. 
  We have $\forall S_0,S_1 \in I_G, S_0 \neq S_1 \implies f(S_0) \neq f(S_1)$. So $f$ is injective. The codomain of $f$ is its range $f(I_G)$, so it is surjective. Therefore, $f$ is bijective.
\end{proof}

Figure~\ref{fig:proof-example} shows how $f$ maps an independent set $S$ of $G$ to an independent set $\bar{S}=f(S)$ of $\bar{G}$ by adding auxiliary nodes into $S$.
Lemma~\ref{lemma:MIS-in-range} is stated as follows, which allows $f$ to connect the solutions of the original problem on $G$ with the solutions of the MIS problem on $\bar{G}$.

\begin{lemma}
  Any MIS of $\bar{G}$ is in the range of function $f$.
  \label{lemma:MIS-in-range}
\end{lemma}

\begin{proof}
  According to the construction of $\bar{G}$,
  $\forall \bar{S}\in f(I_G)$, $\bar{S}\in I_{\bar{G}}$ (i.e., $\bar{S}$ is an independent set of $\bar{G}$), so we have $f(I_G)\subseteq I_{\bar{G}}$.
  Actually, we can prove that $f(I_G)$ is the set of all the independent sets
  $\bar{S}$ of $\bar{G}$ such that 
  $\eri(\bar{S})=\eri(V)$, i.e., $f(I_G) = \{\bar{S} \in I_{\bar{G}} \mid \eri(\bar{S})=\eri(V)\}$.
  
  $\subseteq$: given $\bar{S} \in f(I_G)$, there exists $S \in I_G$ such that $f(S)=\bar{S}$. 
  The function $f$ augments 
  the independent set $S$ with $\bar{v}_k$  for all $k$ values that are not in $\eri(S)$, thus
  we have $\eri(f(S))=\eri(\bar{S})=\eri(V)$.
  
  $\supseteq$: given an independent set
  $\bar{S}\in I_{\bar{G}}$ such that 
  $\eri(\bar{S})=\eri(V)$,
  we have $S = \bar{S} \setminus V'$ is 
  an independent set of $G$
  and $\bar{S} = f(S)\in f(I_G)$.
  
  It's easy to see that 
  a MIS $\bar{S}$ of $\bar{G}$
  satisfying $\eri(\bar{S})=\eri(V)$ (otherwise some $\bar{v}_k$ could be added without
  breaking independence, contradicting maximality).
  Thus, $f(I_G)$ contains all MIS of $\bar{G}$.
\end{proof}

\begin{proof}[Proof of Theorem~\ref{theorem:MIS}]
    
  Define objective functions 
  \[ F(S)=|S| + |\eri(V)| - |\eri(S)|, \quad \text{and} \quad \bar{F}(S)=|S|. \] 
  Since $|\eri(V)|$ is a constant given $G$, optimizing $F(S)$ is equivalent to optimizing $F(S) - |\eri(V)|$ which is the objective function in Equation~\ref{eq:ori-obj}, as additive constants do not affect the optimal solution. We now prove that 
  solving the original problem defined by Equation~\ref{eq:ori-obj} ($\argmax_{S \in I_G} F(S)$) on graph $G$ is equivalent to solving the MIS problem ($\argmax_{\bar{S} \in I_{\bar{G}}} \bar{F}(\bar{S})$) on the auxiliary graph $\bar{G}$: 
  
  $f(S)$ adds auxiliary nodes with disappeared $\eri$ values into $S$. So we have
  \begin{equation}
      \begin{aligned}
          \forall S \in I_G,\bar{F}(f(S)) &= |f(S)| \\ 
          &=|S| + |\eri(V)| - |\eri(S)| \\
          &=F(S).
      \end{aligned}
      \label{eq:F-f-trans}
  \end{equation}
  As proved in Lemma~\ref{lemma:bijective}, $f$ is bijective, so we can rewrite Equation~\ref{eq:F-f-trans} as
  \begin{equation}
      \forall \bar{S} \in f(I_G), \bar{F}(\bar{S})=F(f^{-1}(\bar{S})).
      \label{eq:F-f-inverse-trans}
  \end{equation}

  Since Lemma~\ref{lemma:MIS-in-range} shows that any MIS of $\bar{G}$ belongs to $f(I_G)$, the objective function for the MIS problem on $\bar{G}$ is equivalent to $\argmax_{\bar{S} \in f(I_G)} \bar{F}(\bar{S})$. Let $\bar{S}^*$ denotes one of the solutions for the MIS problem on $\bar{G}$ and $S^*=f^{-1}(\bar{S}^*)$. Then we have
  \begin{equation*}
      \begin{aligned}
          F(S^*) &= F(f^{-1}(\bar{S}^*)) &&\\ 
          &= \bar{F}(\bar{S}^*) &&\text{(by Equation~\ref{eq:F-f-inverse-trans})} \\ 
          &= \max_{\bar{S} \in f(I_G)} \bar{F}(\bar{S}) &&\text{(because $\bar{S}^*$ is an MIS on $\bar{G}$)} \\
          &= \max_{S \in I_G} \bar{F}(f(S)) &&\text{(replace variable $\bar{S}$ with $f(S)$)} \\
          &= \max_{S \in I_G} F(S) &&\text{(by Equation~\ref{eq:F-f-trans})} .
      \end{aligned}
  \end{equation*}
  So $S^*=\argmax_{S \in I_G} F(S)$.
  Thus, any solution $\bar{S}^*$ for the MIS problem on $\bar{G}$ yields a solution $S^*=f^{-1}(\bar{S}^*)$ for the optimization problem on $G$ defined by Equation~\ref{eq:ori-obj}.
\end{proof}

\end{document}